%% file: main.tex
\documentclass{article}
\usepackage[T1]{fontenc}
\usepackage{lmodern}
\usepackage[utf8]{inputenc}
\usepackage{authblk}
\usepackage{hyperref}
\usepackage[numbers]{natbib}
\usepackage{doi}
\usepackage{todonotes}
\usepackage{amsmath}
\usepackage{amsthm}
\usepackage{amsfonts}
\usepackage{mathtools}
\usepackage[ruled, linesnumbered, vlined]{algorithm2e} %noend instead of vlined is more compact, but uglier
\usepackage{cleveref}
\usepackage{tikz}
\usepackage{longtable}
\usepackage{orcidlink}
\usepackage{svg}
\usetikzlibrary{arrows.meta,positioning}

\usepackage[a4paper,top=1cm,bottom=3cm,left=3cm,right=3cm,marginparwidth=4.5cm]{geometry}

\usepackage{bm}

\newtheorem{theorem}{Theorem}
\newtheorem{lemma}{Lemma}
\newtheorem{corollary}{Corollary}
\newtheorem{proposition}{Proposition}
\newtheorem{observation}{Observation}

\theoremstyle{definition}
\newtheorem{definition}{Definition}

\usepackage{enumitem}
\setlist[description]{leftmargin=\parindent,labelindent=\parindent}

\newcommand{\U}[1][]{
\ifthenelse{\equal{#1}{}}{\mathcal{U}}{\mathcal{U}_{#1}}%
}
\newcommand{\N}[1][]{
\ifthenelse{\equal{#1}{}}{\mathcal{N}}{\mathcal{N}_{#1}}%
}
\newcommand{\NP}[1][]{
\ifthenelse{\equal{#1}{}}{\mathcal{N}^{()}}{\mathcal{N}^{()}_{#1}}%
}
\newcommand{\Ui}[1][]{
\ifthenelse{\equal{#1}{}}{\dot{\mathcal{U}}}{\dot{\mathcal{U}}_{#1}}%
}
\newcommand{\Ni}[1][]{
\ifthenelse{\equal{#1}{}}{\mathcal{N}^l}{\mathcal{N}^l_{#1}}%
}

\newcommand{\NiP}[1][]{
\ifthenelse{\equal{#1}{}}{\dot{\mathcal{N}}^{()}}{\dot{\mathcal{N}}^{()}_{#1}}%
}

\newcommand{\PP}{\mathbb{P}}

\newcommand{\move}[1]{\mathrel{\raisebox{-2pt}{$\xrightarrow{#1}$}}}
\newcommand{\isom}[1][]{
\ifthenelse{\equal{#1}{}}{\simeq}{\mathrel{\raisebox{-2pt}{$\overset{#1}{\simeq}$}}}%
}

\DeclareMathOperator{\Aut}{Aut}
\DeclareMathOperator{\Rep}{Rep}
\DeclareMathOperator{\prior}{prior}
\DeclareMathOperator{\likelihood}{likelihood}
\DeclareMathOperator{\id}{id}

\DeclareMathOperator{\VM}{Vert^{--}}
\DeclareMathOperator{\VP}{Vert^+}

\newcommand{\Tail}{\mathrm{Tail}}

\newcommand{\Head}{\mathrm{Head}}

\newcommand{\rSPR}{\mathrm{rSPR}}

\newcommand{\rNNI}{\mathrm{rNNI}}

\title{Metropolis-Hastings Sampling of Phylogenetic Networks: Correcting for Symmetries}

\author[1]{Leo van Iersel \orcidlink{0000-0001-7142-4706}}
\author[2]{Remie Janssen \orcidlink{0000-0002-5192-1470}}
\author[3]{Mark Jones}
\author[4]{Yukihiro Murakami \orcidlink{0000-0003-1355-5884}}
\author[5]{Christopher Reichling}
\affil[1]{TU Delft, The Netherlands, l.j.j.vanIersel@tudelft.nl}
\affil[2]{National Institute for Public Health and the Environment, The Netherlands, remie.janssen@rivm.nl}
\affil[3]{Middlesex University London, United Kingdom, m.jones@mdx.ac.uk}
\affil[4]{TU Delft, The Netherlands, y.murakami@tudelft.nl}
\affil[5]{TU Delft, The Netherlands, christopher.z.a.reichling@gmail.com }

\date{\today}

\begin{document}

\maketitle
\begin{abstract}
In phylogenetics, Metropolis-Hastings methods are commonly used to sample phylogenetic trees or networks, for example from Bayesian posteriors.
These methods generally use transitions that distinguish all nodes involved, and thus require fully labelled representations of phylogenetic networks.
We argue that sampling leaf-labelled phylogenetic networks demands a correction for the number of fully labelled representatives of a leaf-labelled network, or, equivalently, for its internal symmetry. Without correction, there is a danger of undersampling networks with internal symmetries.
We show that this correction can be realized by a quotient construction on the Metropolis-Hastings Markov chain, which, in practice, requires the calculation of the size of the network's automorphism group.
Using $\mu$-vectors, we show that the automorphism group is trivial for orchard networks, and thus also for tree-child networks and trees.
This implies that a correction for symmetry is not needed when sampling only from such network classes.
More generally, using our Python implementation of the algorithms in this paper, we show that using $\mu$-vectors can significantly speed up calculations of automorphism group sizes and thus of Metropolis-Hastings sampling of leaf-labelled networks.
\end{abstract}

% \begin{keyword}
% phylogenetic networks \and Metropolis-Hastings \and quotient Markov chain \and lumpable Markov chains \and graph isomorphism
% \end{keyword}

\section{Introduction}
Phylogenetic networks are directed graphs that are used in biology to represent evolutionary histories \cite{bapteste2013networks}, which are central to much modern biological research \cite{arnason2018whole,forster2020phylogenetic}. Accurate representations are reconstructed using (mostly) genetic data, and the number of methods that can be used to this end is still growing. These methods can be split into two categories: combinatorial methods, and evolutionary model based methods. 

Combinatorial methods could, for example, consist of combining incompatible trees into a network \cite{yang2013quartet,bordewich2007computing,van2019practical,markin2019robinson} or small networks into larger networks \cite{van2014trinets,oldman2016trilonet,huber2017reconstructing,huber2018quarnet,gross2020distinguishing}, or they could build networks directly based on data \cite{bordewich2016determining,van2020reconstructibility,huber2021reconstructibility}. Evolutionary model based methods are generally based on a statistical model which describes the relation between a network with numerical parameters, such as branch lengths and inheritance probabilities, and (a distribution of) genetic data that we would find if the real evolutionary history could be represented by said network. The reconstruction method then either finds a network that best represents the data, e.g., with maximum likelihood, \cite{jin2006maximum,yu2014maximum,yu2015maximum,solis2016inferring,wu2020inference}, or it is used to update our beliefs about the likelihood of the networks in a Bayesian fashion \cite{vaughan2017inferring,wen2016bayesian,zhang2018bayesian}.

\subsection{Network search}
The computational problems involved in many of these methods are quite hard to solve \cite{foulds1982steiner,chor2006finding,roch2006short,bryant2017quirks}. Hence, it is common to employ heuristics. In particular, maximum likelihood methods use local search heuristics (e.g., \cite{yu2014maximum,wen2016reticulate}), and Bayesian methods use Markov chain Monte Carlo (MCMC) methods \cite{vaughan2017inferring,wen2016bayesian,zhang2018bayesian}. Both local search and MCMC heuristics traverse a space of phylogenetic networks by starting at one network and repeatedly making small changes to the network. 

These changes are called rearrangement moves, and several types are commonly used.
Moves that do not change the number of reticulation events (which model, e.g., hybridization or lateral gene transfer) are called horizontal moves, and examples are $\rSPR$ moves, $\Tail$ moves, $\Head$ moves, $\rNNI$ moves, and (part of) SNPR moves \cite{gambette2017rearrangement,bordewich2017lost,janssen2018exploring}. Moves that do change the number of reticulation events are called vertical moves. 
These typically add an edge to the network by subdividing two edges with two new nodes and attaching the new edge between these new nodes, or they remove an edge using the reverse of this process \cite{gambette2017rearrangement}. 

For the heuristics to be effective, it is paramount that the search space is connected. 
It has previously been shown that this is indeed the case for most rearrangement moves \cite{bordewich2017lost,janssen2018exploring,janssen2018heading,thesis_janssen}. 
More specifically, these results imply that, for most rearrangement moves, spaces of leaf-labelled and fully labelled networks are connected.

\subsection{Bayesian Methods}
In Bayesian methods for phylogenetics, prior beliefs about the likelihood $\prior(N)$ of each network are updated on the basis of data $\theta$. This gives new beliefs in the form of a posterior distribution
\[\pi(N|\theta)=\frac{\prior(N)\likelihood(\theta|N)}{\PP(\theta)}.\] 
In general, this distribution cannot be calculated directly, and must thus be approximated. This can be done by sampling from this posterior distribution using, for example, a Markov chain Monte Carlo (MCMC) method. The Markov chain in this method is designed to have $\pi(N|\theta)$ as its stationary distribution, so that sampling using a Monte Carlo method will reflect this posterior distribution.

One way to design a suitable Markov chain, is by using a Metropolis-Hastings method. This method requires the definition of a proposal distribution $g(N_{t+1}|N_t)$ (a method to randomly choose a next state $N_{t+1}$) and a way to calculate an acceptance probability $A(N_{t+1}|N_t)$. A common choice for $A$ that guarantees the desired stationary distribution is the Metropolis ratio
\[A(N_{t+1}|N_t)=\min\left(1,\frac{\pi(N_{t+1}|\theta)g(N_{t}|N_{t+1})}{\pi(N_t|\theta)g(N_{t+1}|N_t)}\right).\]
Calculating this acceptance ratio then requires the (independent) calculation of 
\[\frac{\pi(N_{t+1}|\theta)}{\pi(N_t|\theta)} = \frac{\prior(N_{t+1})\likelihood(\theta|N_{t+1})}{\prior(N_t)\likelihood(\theta|N_t)},\]
and of the Hastings ratio 
\[\frac{g(N_{t}|N_{t+1})}{g(N_{t+1}|N_t)}.\]
The former is the domain of evolutionary models. The model used to calculate the likelihood is often coalescence based, and the prior can for example be based on some assumptions for the parameters of the network (Phylonet \cite{wen2018inferring}) or on a birth-hybridization model for phylogenetic networks (BEAST 2: SpeciesNetwork \cite{zhang2018bayesian}). 

When calculating the Hastings ratio (i.e., the fraction involving the proposal distributions), one cannot only consider the numerical parameters of the network (e.g., branch lengths and inheritance probabilities). One must also take into account the network topology. This is because, even though the sampling is done for leaf-labelled networks, the proposal distribution is generally defined for fully labelled networks. Consider, for example, the networks~$N$ and~$M$ in Figure~\ref{fig:AsymmetricProposal} and suppose we uniformly sample from fully labelled networks using an MCMC method as described above. If we fully label $N$ and $M$ with a set of three tree node labels and two reticulation node labels, then $N$ has 6 corresponding fully labelled networks, whereas $M$ has 12. This results in representatives of $M$ being sampled roughly twice as much as representatives of $N$. Indeed, in a simple experiment, $N$ was sampled 255 times and the two versions of $M$ were sampled 587 and 575 times (Figure~\ref{fig:undersampling_n2k2}, Appendix~\ref{sec:Appendix_Sampling}).

\begin{figure}[ht]
    \centering
    \includegraphics[width=1.0\textwidth]{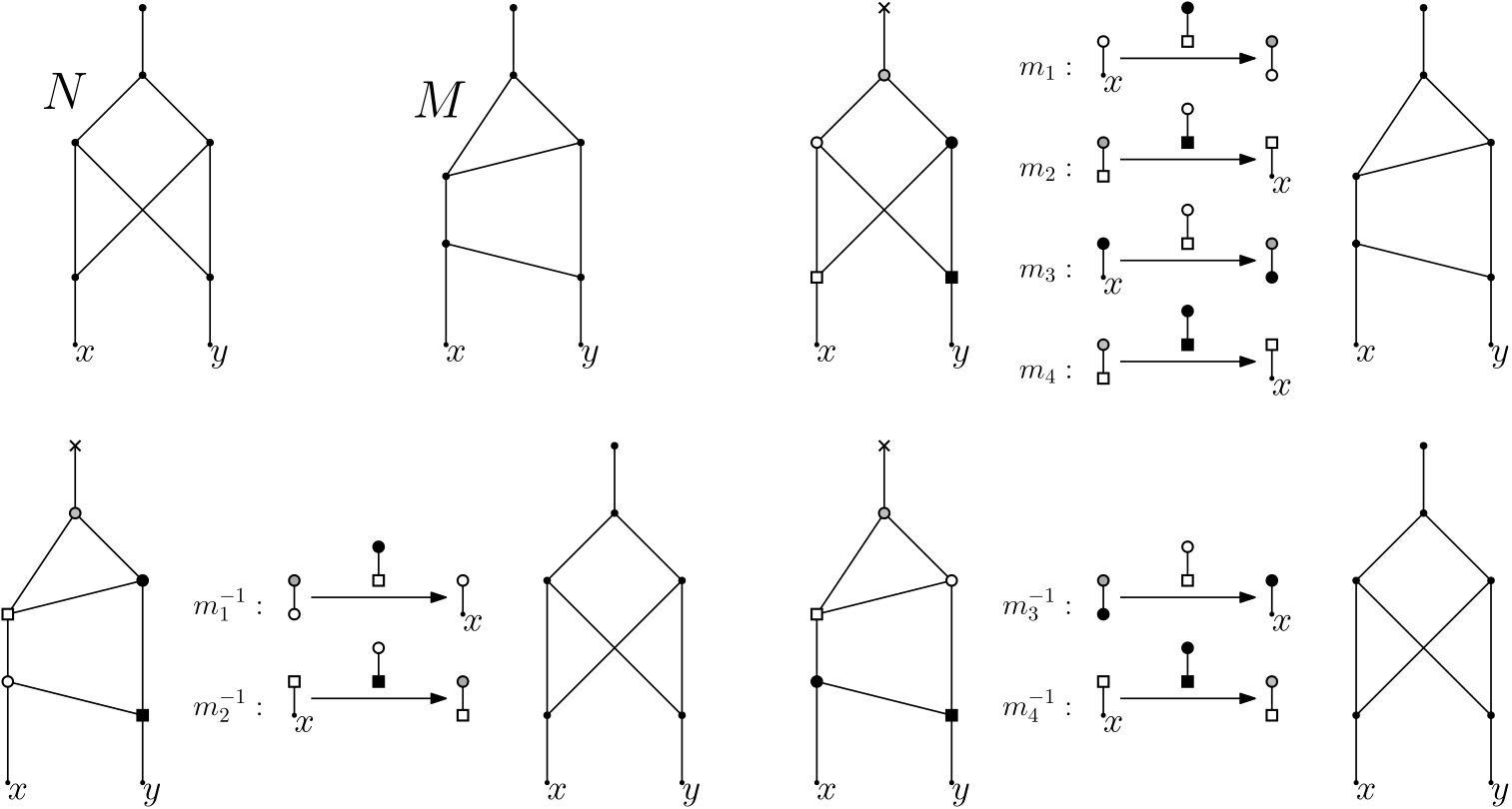}
    \caption{Asymmetric proposals for networks with two leaves and two reticulations. 
    Top left: Two networks $N$ and $M$.
    Top right: There are four $\rSPR$ moves (defined in Section~\ref{sec:rearr}) that change any fully labelled version of $N$ into a network with leaf-labelled topology of $M$.
    Bottom: However, for each fully labelled version of $M$, there are only two moves that change it into a network with leaf-labelled topology of $N$.
    This means that for the leaf-labelled networks $N$ and $M$, the proposal $N\to M$ and its inverse $M\to N$ do not have the same probability.
    Moves are denoted using the arrow notation defined in Section~\ref{sec:rearr}, where edges are shown as they appear in the networks.}
    \label{fig:AsymmetricProposal}
\end{figure}

In general, this results in undersampling of networks that have non-trivial symmetries, as these networks have fewer fully labelled representations than networks without such symmetries. Another way to see this is that moves are uniquely reversible in the space of fully labelled networks, but not in the space of leaf-labelled networks (Figure~\ref{fig:AsymmetricProposal}).
By considering both a Markov chain on leaf-labelled networks and on fully labelled networks, whose relation is that of a Quotient Markov chain \cite{mitavskiy2008quotients}, we correct for this undersampling. 
This paper, hence, proposes a way to sample leaf-labelled networks by correcting for symmetries.

\begin{figure}
    \centering
    \includegraphics[width=.6\textwidth]{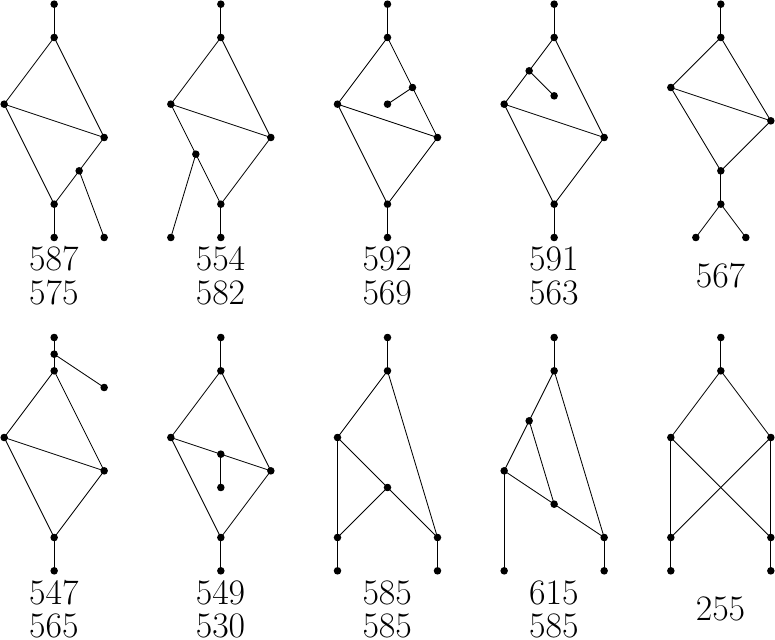}
    \caption{Number of samples taken from each network in $\N(2,2)$---the set of all leaf-labelled networks, with two leaves and two reticulation nodes---using a Metropolis Hastings method without correction for symmetries. The number below a network topology is the number of samples of that network. If there are two numbers below a network topology, they represent the two leaf-labellings of the network topology. Note that the symmetric network on the bottom right is sampled roughly half as often as the other networks. We have sampled 10000 networks with 100 move proposals between each sample, of which roughly half were rejected. See Appendix~\ref{sec:Appendix_Sampling} for details regarding the sampling.}
    \label{fig:undersampling_n2k2}
\end{figure}

\subsection{Our Contribution}
We propose a method to sample from leaf-labelled topologies using a quotient Metropolis-Hastings Markov chain on the space of fully labelled phylogenetic networks.
We work out Metropolis-Hastings for quotient Markov chains in general (Theorem~\ref{the:quotient_MH}), and then apply this theory to Metropolis-Hastings sampling of leaf-labelled phylogenetic networks.
The specific Metropolis-Hastings acceptance probabilities for leaf-labelled networks are proven in Theorem~\ref{the:MH_for_leaf_labelled_networks}. 
These acceptance probabilities include a correction for symmetries of leaf-labelled networks, which requires computation of the number of automorphisms. We give an adapted algorithm (Algorithm~\ref{alg:AUTOM}) that achieves this for phylogenetic networks in general, and an improved version of this algorithm that uses $\mu$-vectors (Algorithm~\ref{alg:automorphism_count}).
By testing our Python implementation of these algorithms, we show that $\mu$-vectors significantly speed up the calculation of the number of automorphisms (Appendix~\ref{sec:Appendix_Running_Time_Mu}).
Moreover, we prove that employing this algorithm is not necessary when sampling from the class of orchard networks (which are basically trees with additional horizontal arcs, see~ \cite{van2022orchard}), as they do not have any such symmetries (Theorem~\ref{the:OrchardTrivAut}). 
Our results can, for example, be used to sample network topologies from (certain classes of) networks uniformly (Appendix~\ref{sec:Appendix_Sampling}).
Finally, by raising questions such as what it means to use a uniform prior, we discuss the choice of prior distributions in Bayesian inference methods for phylogenetic networks.

\section{Preliminaries}
In this section, we provide the definitions used throughout this paper. It consists of two parts, the first two pertaining to random sampling of graphs using Metropolis-Hastings, and the third to phylogenetic networks. This includes an extension of $\rSPR$ moves to fully labelled networks, and some previous results regarding rearrangement moves. 

\subsection{Markov Chains}\label{sec:MCMC}
We will now introduce some notation and basic theory about Markov chains. For simplicity, we restrict ourselves to Markov chains with discrete state spaces, but note that all theory also applies to Markov chains with continuous state spaces, which have transition kernels instead of transition matrices.

\begin{definition}
A Markov chain on a state space $X$ is \emph{irreducible} if there is a sequence of state transitions of positive probability between any pair of states in $X$.
\end{definition}

To correctly sample from the state space using Markov chains, it is also important that the Markov chain is aperiodic. For a Markov chain with transition matrix $P$ to be aperiodic, it suffices for it to be irreducible, and to have one state $x\in X$ with $P_{x\to x}>0$. We will be free enough in choosing our transition matrix to achieve this latter property by modifying any transition matrix $P$ to $\bar{P}:= (1-\epsilon) P + \epsilon I$.

It is well known that any irreducible aperiodic Markov chain has a unique stationary distribution. On the other hand, given a set of transitions, by carefully constructing a Markov chain we can ensure a stationary distribution with for example Metropolis-Hastings.

\subsubsection{Metropolis-Hastings}\label{sec:MH}
Basics about Metropolis-Hastings are quite easy to find, \cite{doi:10.1080/00031305.1995.10476177} gives a great introduction, for example.
Suppose we are given some reversible set of transitions that connects a state space $X$, and some method to pick a transition so that the quotient of proposal probabilities $g(m:x\to y)/g(m^{-1}:y\to x)$ can easily be computed for any transition $m$. Suppose, moreover, that we have a distribution $\pi$ on $X$ with $\pi(x)>0$ for all $x\in X$, such that we can compute $\pi(x)$ for all $x\in X$. Then, using Metropolis-Hastings, we can set up a Markov chain on $X$ with stationary distribution $\pi$ by suitably correcting the acceptance probability for proposal probabilities 
\[A(m) = \min\left(1, \frac{\pi(m(x))}{\pi(x)} \frac{g(m^{-1})}{g(m)} \right),\]
where $m$ is a rearrangement move that turns a network $x$ into another network $m(x)$.
This results in a transition matrix with terms $P_{x\to y} = \sum_{m:x\to y} g(m) A(m)$ for $x\neq y$. 
For completeness' sake, we include pseudo-code for Metropolis-Hastings sampling of phylogenetic networks in Algorithm~\ref{alg:phylogeneticMH}. Note that the space of networks induced by the set $\mathcal{N}$ must be connected and finite, and the sample frequency $f$ large enough for the sample to be (nearly) uniform.

\vspace{\baselineskip}
\begin{algorithm}[ht]
 \KwData{A set of networks $\mathcal{N}$, a starting network $N\in\mathcal{N}$, a number of samples $n$, and a sampling frequency $f$.}
 \KwResult{A (nearly) uniform sample $Y$ of $n$ networks (with repetition) from $\mathcal{N}$.}
 Set $Y:=\emptyset$\;
 Set $\pi(x)=1$ if $x\in\mathcal{N}$ and $\pi(x)=0$ otherwise.
 \For{$i=1...n$}{
    \For{$j=1...f$}{
        Generate a move proposal $P$ as described in Section~\ref{sec:proposal_probabilities_worked_out}\;
        \If{$P$ corresponds to a valid move $m:N\to N'$ acting on $N$}{
           Calculate the acceptance probability $A$\;
        }\Else{
           $P$ is invalid, and we set $A=0$\;
        }
        Uniformly at random pick a number $a\in [0,1]$\;
        \If{$a\leq A$}{
            Set $N:=N'$\;
        }
    }
    Add $N$ to the multi-set $Y$\;
}
\Return $Y$\;
\caption{\textsc{Metropolis-Hastings}$(\mathcal{N},N, n, f)$}\label{alg:phylogeneticMH}
\end{algorithm}
\vspace{\baselineskip}

\paragraph{Detailed Balance for Move Dependent Acceptance}
Note that most introductions to Metro\-polis-Hastings define acceptance probabilities for pairs of states, and they use detailed balance (i.e., $\pi(x) P_{x\to y} = \pi(y) P_{y\to x}$) to prove that the stationary distribution is $\pi$. Recall that, indeed, if detailed balance holds we have the following equality for all $x\in X$:
\begin{align*}
    \sum_{y\in X}\pi(y)P_{y\to x} = \sum_{y\in X} \pi(x) P_{x\to y} = \pi(x)\sum_{y\in X} P_{x\to y} = \pi(x).
\end{align*}
In other words, applying the transition matrix $(P_{x\to_y})_{x,y\in X}$ to the vector $\pi$ results in $\pi$ again, so the distribution $\pi$ is the stationary distribution of the Markov chain.

Here, instead of pairs of states, we use a slightly more specific acceptance probability that is a function of a state and a move. The stationary probability is still guaranteed because detailed balance still holds. To see this, first note that for each $m$ that turns $x$ into $y$, we now either have $A(m)\leq 1$, in which case $A(m^{-1})=1$ and 
\[A(m) = \frac{\pi(y)}{\pi(x)} \frac{g(m)}{g(m^{-1})}.\]
Hence, 
\begin{align*}
 \pi(x)g(m)A(m) & = \pi(x)g(m) \frac{\pi(y)}{\pi(x)} \frac{g(m^{-1})}{g(m)} \\
                            & = \pi(y) g(m^{-1}) \\
                            & = \pi(y) g(m^{-1}) A(m^{-1})
\end{align*}

Or, we have $A(m^{-1})\leq 1$, in which case we can similarly deduce that
\[
\pi(x)g(m)A(m)= \pi(y) g(m^{-1}) A(m^{-1}).
\]

Now to return to detailed balance: this is trivial when $x=y$, so suppose $x\neq y$. Then, by the above and the fact that the moves are uniquely reversible, we have detailed balance:
\begin{align*}
 \pi(x) P_{x\to y} & = \pi(x)\sum_{m:x\to y}g(m)A(m) \\
                & = \sum_{m:x\to y}\pi(x) g(m)A(m) \\
                & = \sum_{m:x\to y}\pi(y) g(m^{-1})A(m^{-1}) \\
                & = \sum_{m':y\to x}\pi(y) g(m')A(m') \\
                & = \pi(y) \sum_{m':y\to x} g(m')A(m') \\
                & = \pi(y)P_{y\to x},
\end{align*}

\subsubsection{Quotient Markov Chains}\label{sec:quotient_markov_chains}
We will now work out some properties of quotient Markov chains. 
Such a quotient is essentially a special version of a lumpable Markov chain, for which references can easily be found. 
But for our purposes, we do need the additional properties that the quotient offers, so we do have to work out the details here.

The following theorem is adapted from \cite{mitavskiy2008quotients} Theorem~8 which takes it originally from the appendix of \cite{vose1999simple}. % https://link.springer.com/article/10.1007/s10710-007-9038-6#Fn3
Here, we assume the original Markov chain is a multi-graph with possibly multiple transitions between two states. The equivalence relation is extended to an equivalence relation on the transitions as well. This initially makes the theory slightly more complicated, but the application to phylogenetic networks we will encounter later matches this interpretation more closely.

\begin{definition}\label{def:quotient_markov_chain}
    Given an irreducible aperiodic Markov chain $M=(X,T,P)$ with stationary distribution $\pi$ and an equivalence relation $\sim$ on the states and on the transitions of $M$, which respects the equivalence relation on the states (i.e. $(u\to v)\sim (p \to q)$ then $u\sim p$ and $v\sim q$). Then the quotient Markov chain $M/\sim=(X/\sim, T/\sim, P')$ is defined by the transition probabilities
    \[
       P'_{\widetilde{m}: U\to V} = \frac{1}{\pi(U)} \sum_{(m:u\to v)\in\widetilde{m}}\pi(u)P_m.
    \]
\end{definition}

\begin{theorem}
    Let $\pi$ denote the stationary distribution of an irreducible aperiodic Markov chain $M=(X,T,P)$. Suppose we are given an equivalence relation $\sim$ partitioning the state space $X$ and the transitions $T$ of $M$ where the second equivalence respects the first, i.e. transitions can only be equivalent if they are between the same pair of aggregated states. Then the quotient Markov chain $M/\sim$ with transition probabilities
    \[P'_{\widetilde{m}} = \frac{1}{\pi(U)} \sum_{(m:u\to v)\in\widetilde{m}}\pi(u)P_m\]
    is irreducible and its unique stationary distribution $\widetilde\pi$ is compatible with $\pi$ in the sense that for every $V\in X/\sim$ we have $\widetilde\pi(V)=\pi(V)$.
\end{theorem}
\begin{proof}
    First, we check that the transition probabilities are actually transition probabilities, and the probabilities for all outgoing transitions sum to one for any state $U\in X/\sim$:
    \begin{align*}
    \sum_{\widetilde{m}: U\to \cdot} P'_{\widetilde{m}} 
        &= \sum_{\widetilde{m}: U\to \cdot} \frac{1}{\pi(U)} \sum_{(m:u\to v)\in\widetilde{m}}\pi(u) P_m \\
        &= \frac{1}{\pi(U)} \sum_{\widetilde{m}: U\to \cdot} \sum_{(m:u\to v)\in\widetilde{m}}\pi(u) P_m \\
        &= \frac{1}{\pi(U)} \sum_{u\in U} \sum_{m: u\to \cdot}\pi(u) P_m \\
        &= \frac{1}{\pi(U)} \sum_{u\in U} \pi(u) \sum_{m: u\to \cdot } P_m \\
        &= \frac{1}{\pi(U)} \sum_{u\in U} \pi(u) \\
        &= \frac{1}{\pi(U)} \pi(U)  \\
        &= 1
    \end{align*}

    Because the original Markov chain is aperiodic and irreducible, it is clear that the resulting Markov chain is aperiodic and irreducible as well. Hence, we only have to prove that the stationary distributions are compatible. For this, we simply calculate the stationary distribution for the quotient chain, which has transition probabilities $P'(V|U) := \sum_{\widetilde{m}: U \to V} P'_{\widetilde{m}}$ between two states $V,U\in X/\sim$.
    \begin{align*}
        \sum_{U\in X/\sim} \pi(U) P'(V|U) 
            &= \sum_{U\in X/\sim} \pi(U) \PP(V|U) \\
            &= \sum_{U\in X/\sim} \pi(U) \sum_{\widetilde{m}: U\to V} P'_{\widetilde{m}} \\
            &= \sum_{U\in X/\sim} \pi(U) \sum_{\widetilde{m}: U\to V} \frac{1}{\pi(U)} \sum_{(m:u\to v)\in\widetilde{m}}\pi(u) P_m \\
%            &= \sum_{U\in X/\sim} \pi(U)\frac{1}{\pi(U)} \sum_{\widetilde{m}: U\to V} \sum_{(m:u\to v)\in\widetilde{m}}\pi(u)p_m \\ 
            &= \sum_{U\in X/\sim} \sum_{\widetilde{m}: U\to V} \sum_{(m:u\to v)\in\widetilde{m}}\pi(u) P_m \\     
    \end{align*}
    Note that in the last summation, each transition to $V$ occurs exactly once, as both the states and the transitions are partitioned by $\sim$. Therefore, we can regroup the transitions by origin node and end node instead of equivalence class of transitions, we continue:
    \begin{align*}
        \sum_{U\in X/\sim} \pi(U) \PP(V|U) 
            % &= \sum_{U\in X/\sim} \sum_{u \in U} \sum_{v\in V} \sum_{m:u\to v}\pi(u)p_m \\
            % &= \sum_{v\in V} \sum_{U\in X/\sim} \sum_{u \in U} \sum_{m:u\to v}\pi(u)p_m \\
            &= \sum_{v\in V} \sum_{u \in X} \sum_{m:u\to v}\pi(u) P_m \\
            &= \sum_{v\in V} \sum_{u \in X} \pi(u) \sum_{m:u\to v} P_m \\
            &= \sum_{v\in V} \sum_{u \in X} \pi(u) P(v | u) \\
            &= \sum_{v\in V} \pi (v) \\
            &= \pi(V)
    \end{align*}

\end{proof}

\paragraph{Quotients Induced by Group Actions}
In the special case that the equivalence relation on the state space is induced by a group action that additionally respects the transition probabilities, we can specialise this result even further. Although little mention of quotient Markov chains could be found, symmetries of Markov chains have been studied, for example in relation to mixing times \cite[e.g.][]{doi:10.1137/070689413, doi:10.1080/15427951.2005.10129100}. As fully worked out formulas for Metropolis-Hastings for quotient Markov chains could not be found, we work out these formulas ourselves here.

\begin{definition}
    An action of a group $\Gamma$ on a graph $G$ is a homomorphism $\Gamma \to \Aut(G)$. In other words, it is a
    pair of group actions on the sets $V(G)$ and $E(G)$ that respect the incidence relations on $G$, i.e., $\gamma((u,v)) := (\gamma(u), \gamma(v)) \in E(G)$. 
\end{definition}
For multi-graphs and digraphs, the definitions are analogous, but with explicitly defining the mapping of the multi-edges and the direction of the edges. 

\begin{definition}
    A group action on a discrete Markov chain is an action on the underlying multi-digraph that also respects the transition probabilities on the edges, i.e. if $P_{m} = P_{\gamma(m)}$ for each group element $\gamma$ and each transition $m$.
\end{definition}

\begin{definition}
    Given an irreducible aperiodic Markov chain $M$ with stationary distribution $\pi$ acted upon by a group $\Gamma$. Then the quotient Markov chain $M/\Gamma$ is the quotient Markov chain $M/\sim$ where $x \sim y$ if $x$ and $y$ are in the same orbit (for $x,y$ both states or both transitions).
\end{definition}

The following lemma shows that the quotient Markov chain under a group action behaves as a sort of aggregation of states. The transitions remain largely as they were, with only a correction for the number of transitions in an equivalence class.

\begin{lemma}\label{lem:transitions_group_quotient}
    Given an irreducible aperiodic Markov chain $M$ acted upon by a group $\Gamma$. Then the transition probabilities of $M/\Gamma$ are 
    \[
       P'_{\widetilde{m}} = \frac{|\widetilde{m}|}{|U|}P_m = \sum_{(m:x\to \cdot) \in \widetilde{m}} P_m,
    \]
    where $\widetilde{m}: U\to V$ and $x\in U$ are arbitrary chosen.
    And the aggregated transitions probabilities between states are
    \[
       P'(V|U) = \sum_{y\in V} P(y|x),
    \]
    where $x\in U$ is an arbitrary element.
\end{lemma}
\begin{proof}
    The first can be seen as follows, by using the Definition~\ref{def:quotient_markov_chain} for the transition probabilities, and the fact that the group action respects transition probabilities.
    \begin{equation*}
        P_{\widetilde{m}} 
            := \frac{1}{\pi(U)} \sum_{(m:u\to v)\in\widetilde{m}}\pi(u)P_m
            = \frac{\pi(x)}{\pi(U)} \sum_{(m:u\to v)\in\widetilde{m}}P_m
            = \frac{\pi(x)}{\pi(U)} |\widetilde{m}| P_m
    \end{equation*}
    where $U$ is the origin state of $\widetilde{m}$ and $x\in U$ arbitrary. By symmetry of the Markov chain, each element of $U$ has the same probability in the stationary distribution, so $\pi(U) = |U| \pi(x)$. Combining the equation above with this one, we see that $P'_{\widetilde{m}}=\frac{|\widetilde{m}|}{|U|}P_m$. Finally, the fact that the number of elements $|\widetilde{m}|$ is equal to the number of transitions in $\widetilde{m}$ from one arbitrary origin $u'\in U$ times $|U|$ implies that $\frac{|\widetilde{m}|}{|U|}P_m = \sum_{(m:x\to \cdot) \in \widetilde{m}} P_m$.    

    The formula for the aggregated transition probabilities is a result of the symmetry, which can be seen as follows. A simple rearrangement after writing out the definitions gives
    \begin{equation*}
       P'(V|U) 
            = \sum_{\widetilde{m}:U\to V} P_{\widetilde{m}}
            = \sum_{\widetilde{m}:U\to V} \sum_{(m:x\to \cdot) \in \widetilde{m}} P_m
    \end{equation*}
    Now note that the double sum simply runs over all transitions $m$ from $x$ to all elements $v\in V$. By first summing over all $v\in V$, and then over all the transitions from $x$ to this $v$, we continue:
    \begin{equation*}            
       \sum_{\widetilde{m}:U\to V} \sum_{(m:x\to \cdot) \in \widetilde{m}} P_m
            = \sum_{y\in V} \sum_{m:x\to y} P_m 
            = \sum_{y\in V} P(y|x),
    \end{equation*}
    which proves the result.
\end{proof}

\subsubsection{Metropolis-Hastings on Quotient Spaces}
In the setting of this paper, we will encounter the following setting. Instead of sampling from $X$, we want to sample from $X/\Gamma$ with distribution $\pi$. Unaware of the difference between $X$ and $X/\Gamma$, one might naively sample from $X$ with distribution $\pi'(u):= \pi(U)$ and marginalise. However, as should be clear by now, this will lead to under-representation of equivalence classes $U\in X/\Gamma$ of relatively small cardinality when viewed as of subsets of $X$.

Correcting for the under-representation of these equivalence classes after sampling can affect sample sizes, which we may want to fix beforehand. Hence, we need to correct for this under-representation during the Metropolis-Hastings sampling. We assume that $G/\Gamma$ is irreducible and aperiodic, that each transition in $G$ is uniquely reversible, and that detailed balance holds for $G$ on the level of transitions. There are two ways of seeing how this can be done, by correcting the stationary distribution directly, or by building a Metropolis-Hastings chain on $G/\Gamma$.

\paragraph{Correcting the Stationary Distribution} 
The first is by adjusting the stationary distribution, which involves counting the number of elements in each equivalence class and changing the stationary distribution accordingly. In other words, we would like to have 
\[
   \widetilde{\pi}(U) = \pi(x') = \frac{1}{|U|}\sum_{x\in U}\pi(x) = \sum_{x\in U}\frac{\pi(x)}{|U|},
\]
where $x'\in U$ is an arbitrary representative, so we sample from $V(G)$ with stationary distribution $\pi'(x):= \frac{\pi(x)}{|U|}$. At first glace it seems this may result in issues when $G$ consists of multiple components, as these components may have different number of representatives of each $v\in V(G)/\Gamma$. However, by symmetry of $G$, this cannot be the case, because the relative sizes of the equivalence classes do not depend on the number of components of $G$. This correction for the stationary distribution leads to a correction factor of $|U|/|V|$ in the transition probabilities of the Metropolis-Hastings Markov chain.

\paragraph{Correcting the M-H Ratio} 
The second correction method goes by building a Metropolis-Hastings Markov chain on the proposal Markov chain $G/\Gamma$ directly. With this, we mean that we calculate the Hastings ratio for $G/\Gamma$, instead of for $G$. Like before, we assume that the transitions in $G$ are reversible, but we calculate the acceptance ratio for the orbit of a transition instead of a single transition. Note that such a transition-orbit $(\widetilde{m}:U\to V)\in E(G)/\Gamma$ still has a unique inverse transition-orbit $(\widetilde{m}^{-1}: V\to U)\in E(G)/\Gamma$, so we can define a Metropolis-Hastings chain on the level of orbit-transitions like in Section~\ref{sec:MH}.

To calculate the Hastings-ratio, we denote with $P_{\cdot}$ the transition probabilities in the Metropolis-Hastings Markov chains and with $g$ the proposal probabilities in $G$ and $G/\Gamma$. By Lemma~\ref{lem:transitions_group_quotient}, we can calculate the Metropolis-Hastings ratio for $G/\Gamma$ as follows:
\[
\frac{g(\widetilde{m}^{-1})}{g(\widetilde{m})}
    = \frac{\frac{|\widetilde{m}^{-1}|}{|V|}g(m^{-1})}{\frac{|\widetilde{m}|}{|U|}g(m)}
    % = \frac{\frac{1}{|V|}g(m^{-1})}{\frac{1}{|U|}g(m)}
    % = \frac{|U| g(m^{-1})}{|V| g(m)}
    = \frac{|U|}{|V|} \frac{g(m^{-1})}{g(m)}.
\]
for a move $\widetilde{m}: U \to V$ with $m\in \widetilde{m}$ and $m^{-1}\in \widetilde{m}^{-1}$ arbitrary representatives. Note that, again, this leads to a correction of $|U|/|V|$ compared to the Metropolis-Hastings chain on $G$.

\paragraph{How to Correct For the Graph Quotient} 
Both of these methods show that the correction for equivalence class size can be realised without knowing whether $G$ is connected, and requires a correction factor $\frac{|U|}{|V|}$ in the acceptance ratio compared to the Metropolis-Hastings chain on $G$.

\begin{theorem}\label{the:quotient_MH}
    Let $G=(V,T,g)$ be a transition graph with probabilities $g(m)$ and let $\pi$ be a distribution on $V(G)/\Gamma$. Then a Metropolis-Hastings Markov chain on $G/\Gamma$ with stationary distribution $\pi$ can be realised by taking the acceptance function
    \[
      A(\widetilde{m}:U\to V) = 
        \min\left(
          1, 
          \frac{\pi(V)}{\pi(U)}\frac{g(\widetilde{m}^{-1})}{g(\widetilde{m})} 
        \right) = 
        \min\left(
          1, 
          \frac{\pi(V)}{\pi(U)}\frac{|U|}{|V|}\frac{g(m^{-1})}{g(m)} 
        \right)
    \]
    where $\widetilde{m}:U \to V$ is a transition in $G/\Gamma$, and $m\in \widetilde{m}$ and $m^{-1}\in \widetilde{m}^{-1}$ are arbitrary representatives of this transition and its reverse.
\end{theorem}

%Note that such an object can alternatively be viewed as an equivalence class whose elements have a relatively large stabilizer subgroup. In other words, they have a larger set of non-trivial symmetries in $G$. It will be of value to remember this view in the setting of phylogenetic networks, where the set of non-trivial symmetries will correspond to the automorphism group of the network.

\subsubsection{Properties of the Quotient Chains}
\paragraph{Irreducibility of the Quotient Chain}
As we have noted before, the proposal Markov chain $G$ does not need to be irreducible for Metropolis-Hastings sampling with the correction from Theorem~\ref{the:quotient_MH} to work on $G/\Gamma$. It is enough for the quotient Markov chain $G/\Gamma$ itself to be irreducible. This is useful, because we do not need to prove that $G$ is irreducible, even though we generate transition proposals on the level of $G$. Moreover, irreducibility is often proven for $G/\Gamma$ because it is simpler, or simply because it is the object of interest. Indeed, this is very often the case for spaces of phylogenetic networks \cite[e.g.][]{huber2016spaces, gambette2017rearrangement, janssen2019rearrangement, thesis_janssen, klawitter2020spaces, ERDOS2021205}. The proof that irreducibility of $G/\Gamma$ is sufficient follows from a lifting property.

\begin{definition}
    A lift of a subgraph $H$ of $G/\Gamma$ to $G$ is a homomorphism of graphs $l:H\to G$ such that $l\circ \pi = \mathrm{id}$, where $\pi: G\to G/\Gamma$ is the projection. The subgraph obtained as the image $l(H)$ of the lift is also called a lift of $H$.
\end{definition}

The following lemma implies that we can lift simple paths from $G/\Gamma$ to $G$. This then implies that the components of $G$ are irreducible if $G$ has reversible transitions and $G/\Gamma$ is irreducible. For simplicity, we note that a directed transition graph in which every transition is reversible can be viewed as an undirected graph (ignoring the transition probabilities). Moreover, the transition graph is irreducible if and only if the undirected graph is connected.

\begin{lemma}
    Let $\Gamma$ be a group acting on an undirected graph $G$ with quotient graph $G/\Gamma$ and projection $\pi: G \to G/\Gamma$. Let $T$ be an arbitrary subtree in $G/\Gamma$, and $u\in p\in T$ be any chosen representative of some node of $T$. Then there exists a lift of $T$ in $G$ that includes $u$.
\end{lemma}
\begin{proof}
    To construct a lift, start with a subgraph $H$ of $G$ consisting only of $u$ and extend $H$ to a lift of $T$ one edge at a time. More precisely, suppose there is an edge $(U,V)\in G/\Gamma$ with $u'\in U$ already in $H$. Because $(U,V)\in G/\Gamma$, there is an edge $(u'',v)\in E(G)$ which represents $(U,V)$. Moreover, there is an element $\gamma \in \Gamma$ that sends $u''$ to $u'$, as these nodes are both in $U$. Applying $\gamma$ to the edge $(u'',v)$, we get an edge $(u',v')\in E(G)$ with which we can extend $H$.   
\end{proof}

\begin{corollary}
    Let $\Gamma$ be a group acting on an undirected graph $G$ with quotient $G/\Gamma$ and projection $\pi: G \to G/\Gamma$. If $G/\Gamma$ is connected, then for any connected component $H$ of $G$ and for all $U\in V(G/\Gamma)$, we have $V(H)\cap U \neq \emptyset$.
\end{corollary}

Hence, if the transition graph $G$ is reducible, then any connected component $H$ of $G$ is still irreducible, and can be used for sampling. Moreover $G/\Gamma \equiv H/\Gamma$, so sampling from $G/\Gamma$ can be done using $H$ via the quotient construction. This makes it sufficient to prove that $G/\Gamma$ is irreducible for sampling from this quotient graph using the original graph, even if that original graph is reducible.

\paragraph{Aperiodicity of the Quotient Chain}
Unlike irreducibility, aperiodicity is not necessarily conserved by a lift of the quotient construction. A trivial example is the Markov chain on a directed cycle on $n$ nodes, where each node has a unique outgoing transition to the next state (which is taken with probability 1). This a Markov chain with period $n$. The cyclic group on $n$ elements acts on this Markov chain such that the quotient is a Markov chain with one state, which is aperiodic. 

If the aim is to have a Markov chain with a given stationary distribution, then aperiodicity can be introduced in the lift without changing the stationary distribution of the quotient chain. If one simply adds a non-zero probability for each state to transition to itself (the same for each node), then the Markov chain becomes aperiodic, but the stationary distribution of the quotient chain does not change.

\subsection{Sampling Graphs}
To sample graphs using this framework, we must realize that in most software implementations, all nodes of a graph are distinguishable objects. Hence, to sample partially labelled or unlabelled graphs, the nodes must be made indistinguishable somehow. In light of the previous discussion, it is natural to start with labelled graphs and use a quotient construction to remove the labels again. In fact, spaces of partially labelled graphs can be viewed as quotient spaces of fully labelled graphs, where the group action on the set of fully labelled graphs is induced by a group action on the label set. Let $S_{\mathcal{L}}$ be the permutation group for the label set $\mathcal{L}$, and let $K\subseteq \mathcal{L}$ be a subset of the labels --- think of these as the leaf labels of phylogenetic networks. Then a quotient by the action of the permutation group $S_{\mathcal{L}\setminus K}$ preserves the labels in $K$, but makes all other labels indistinguishable.

This means that we can sample unlabelled or partially labelled graphs using a quotient construction as described above. 
Provided we have a transition graph on the set of all fully labelled graphs of interest, for which the quotient is an irreducible transition graph.

\subsubsection{Rearrangement Moves}\label{sec:rearrangement_move_proposals}
To sample graphs based on the quotient construction above, we need to define transitions on the set of all fully labelled networks. For graphs, we call these transitions rearrangement moves or rearrangement operations, which is common for phylogenetic networks. For different types of graphs, there are different types of rearrangement moves. For example, edge-switches are often used for graphs with a given degree sequence \cite[e.g.][]{GREENHILL20181, kleer2019nash}, and rooted Subtree Prune and Regraft ($\rSPR$) moves are often used for rooted phylogenetic trees \cite[e.g.][]{nascimento2017biologist}.

We will now introduce a way to use such moves to define a transition graph for which the transition probabilities are easy to calculate. For this, we define a \textit{move proposal} as a tuple of nodes and or edges that fully specify a rearrangement move. For example, a switch on a directed graph can be fully specified by a tuple of four nodes $(a,b,c,d)$, which \textit{encodes} a swap that removes the edges $(a,b)$ and $(c,d)$ and adds the edges $(a,d)$ and $(c,b)$. Of course, for a given graph, many of these proposals can be \textit{invalid}, meaning they do not define a swap in that graph. 

To use the Metropolis-Hastings sampling above, it is convenient if these move proposals have unique reverse proposals. For example, the switch encoded by $(a,b,c,d)$ is the reverse of the switch encoded by $(a,d,c,b)$. To assign transition probabilities, one could simply define a method to randomly generate move proposals such that the probability of a given proposal can easily be determined. With such a method, it becomes feasible to quickly generate a proposal and compute its acceptance probability.

For example, a simple (though inefficient) way to generate swap proposals is by uniformly sampling four nodes with replacement from the set of nodes of the graph $G$, the probability of each proposal then becomes $1/|V(G)|^4$, which is easy to compute. In fact, the probability only depends on the (length of the) degree sequence, so each move proposal in the transition graph has the same probability. The downside of this proposal generator is that it produces a lot of invalid proposals.

As mentioned in the example, generating proposals this freely does have the downside of (in general) producing many invalid proposals. This is no problem for the Metropolis-Hastings sampling, as such invalid proposals are simply discarded --- their acceptance probability is set to $0$, or equivalently, $\pi(D)=0$ where $D$ is the result of applying the move proposal. However, invalid proposals result in the need for longer chains, which simply makes sampling slower.

\subsubsection{Correcting for Equivalence Class Sizes}\label{sec:correction_representations}

In the case of the graph action outlined above, the (relative) sizes of the equivalence classes can be determined by the sizes of their automorphism groups. Indeed, let $G$ be a partially labelled (by a label set $K\subseteq \mathcal{L}$) or unlabelled graph (so $K=\emptyset$), and $X$ the set of all fully labelled graphs obtained from $G$ by adding labels from the set $L=\mathcal{L}\setminus K$ on all other nodes. The number of fully labelled graphs obtained representing $G$ is then 
\begin{equation}\label{eq:correction_representation}
    \Rep(G) = \frac{|L|!}{(|L|-|V(G)|+|K|)!} \frac{1}{|\Aut(G)|},
\end{equation} 
where $\Aut(G)$ is the group (set) of automorphisms of $G$ that respects the labelling of $G$. The group action and the quotient are chosen such that the equivalence classes are exactly these sets of fully labelled graphs representing one partially labelled or unlabelled graphs, so this $\Rep(G)$ is the size we need to correct for. In most cases, the rearrangement moves do not change the number of labels or the number of nodes, so the size of the automorphism group is the only factor to correct for.

For phylogenetic networks, we will change this slightly, as we will consider separate label sets for different types of nodes. However, the principle remains the same: there is a correction factor for the number of automorphisms of the network. 

\subsection{Phylogenetic Networks}
\subsubsection{Directed Phylogenetic Networks}
%DEF: Phylogenetic network
\begin{definition}
A \emph{directed (binary) phylogenetic network} $N=(V,A,l)$ on a set of \emph{taxa} $X$ is a directed acyclic graph (DAG) $(V,A)$ labelled bijectively with $l:L(N)\rightarrow X$, where $L(N)$ denotes the set of leaves of $N$, with no parallel edges and only nodes of the following types:
\begin{description}
    \item[Root:] indegree-$0$, outdegree-$1$ node;
    \item[Tree node:] indegree-$1$, outdegree-$2$ node;
    \item[Reticulation:] indegree-$2$, outdegree-$1$ node;
    \item[Leaf:] indegree-$1$, outdegree-$0$ node,
\end{description}
and such that there is exactly one root.
We shall write~$V(N)$ to denote the set of all nodes of~$N$ and~$A(N)$ to denote the set of all edges of~$N$.
The reticulation number $r(N)$ of $N$ is the number of reticulation nodes in $N$. A \emph{phylogenetic tree} is a phylogenetic network without reticulations.
\end{definition}

For brevity, we will refer to directed binary phylogenetic networks simply as networks in the remainder of this paper.

\begin{definition}
A \emph{fully labelled network} (cf. \emph{vertex-labelled network}~\cite{fuchs2019counting}, or \textit{internally labelled network}~\cite{thesis_janssen}) is a network
$\dot{N}=(V,A,l)$ in which the labelling is an injective map $l:V\to X$. The set of labels $X$ can be partitioned into sets that correspond to the types of nodes. The set of leaf labels is $X^l$, the set of tree node labels is $X^t$, the set of reticulation labels is $X^r$, and the set of root labels is the singleton set $\{x^{\rho}\}$.
\end{definition}

To clearly distinguish between networks and fully labelled networks, we sometimes refer to the former as \textit{leaf-labelled networks}. Because each leaf has a unique label in a network, we will often interchange $l(v)$ and $v$ freely for leaves in networks. For fully labelled networks, we do this for all nodes. 

\begin{observation}\label{obs:numbers}
Let $N=(V,A,l)$ be a binary network with $n$ leaves and $k$ reticulations. Then $N$ has $2n+3k-1$ edges and $2n+2k$ vertices, of which $n$ are leaves, $1$ is the root, $k$ are reticulation nodes, and $n+k-1$ are tree nodes.
\end{observation}

\begin{definition}
Two networks are \emph{isomorphic} if there exists a bijection~$f$ between the nodes of~$N$ and the nodes of~$N'$ such that~$(u,v)$ is an edge of~$N$ if and only if~$(f(u),f(v))$ is an edge of~$N'$, and each labelled node of~$N$ is mapped to a node in~$N'$ with the same label. They are \emph{leaf-isomorphic} if they are isomorphic as labelled graphs when the labelling is restricted to the leaves. 
An \emph{automorphism} of a network $N$ is an isomorphism between $N$ and itself. The automorphism group of $N$ is denoted $\Aut(N)$.
\end{definition}

\subsubsection{Rearrangement Moves}\label{sec:rearr}
In this section, we recall the definitions of rearrangement moves on phylogenetic networks, including an extension to networks with internal labels. These moves provide a way to slightly modify networks. We will define several of these moves, and present a few basic properties of these moves. It is conventional \cite[e.g.][]{gambette2017rearrangement} to make a distinction between horizontal and vertical moves, where moves that do not change the reticulation number are called \emph{horizontal moves} and moves that do change the reticulation number are called \emph{vertical moves}. We first introduce some horizontal moves, and then we continue with vertical moves.

\paragraph{Horizontal Moves}

%DEF: Tail move
\begin{definition}\label{def:tail}
Let $N$ be a network containing the distinct edges $(p,u)$, $(u,c)$, $(u,v)$, and $(p',c')$. 
Let $D$ be the (multi-)digraph obtained from $N$ by 
\begin{itemize}
    \item \emph{pruning} $(u,v)$ at $u$, i.e., replace $(p,u)$ and $(u,c)$ with a new edge $(p,c)$;
    \item and \emph{reattach} $(u,v)$ at $(p',c')$, i.e., replace $(p',c')$ by $(p',u)$ and $(u,c')$.
\end{itemize}
If $D$ is a network (i.e. if it has no parallel edges or directed cycles), then we say $D$ is the result of the \emph{tail move} of $(u,v)$ from $(p,c)$ to $(p',c')$ in $N$, which we also denote $(p,c)\move{(u,v)}(p',c')$ or $u\move{(u,v)}(p',c')$ if the from-edge is not relevant. All nodes retain their labels if they were labelled (Figure~\ref{fig:InternalLabelsExample}).
\end{definition}

A head move is defined similarly, with the following small change. Instead of the tail, the head of an edge is pruned and reattached in a head move.

%DEF: Head move
\begin{definition}\label{def:head}
Let $N$ be a network containing the distinct edges $(p,v)$, $(v,c)$, $(u,v)$, and $(p',c')$. 
Let $D$ be the (multi-)digraph obtained from $N$ by 
\begin{itemize}
    \item \emph{pruning} $(u,v)$ at $v$, i.e., replace $(p,v)$ and $(v,c)$ with a new edge $(p,c)$;
    \item and \emph{reattach} $(u,v)$ at $(p',c')$, i.e., replace $(p',c')$ by $(p',v)$ and $(v,c')$.
\end{itemize}
If $D$ is a network (i.e. if it has no parallel edges or directed cycles), then we say $D$ is the result of the \emph{head move} of $(u,v)$ from $(p,c)$ to $(p',c')$ in $N$, which we also denote $(p,c)\move{(u,v)}(p',c')$ or $v\move{(u,v)}(p',c')$ if the from-edge is not relevant.
All nodes retain their labels if they were labelled (Figure~\ref{fig:InternalLabelsExample}). 
\end{definition}

%TEXT: after the move, it needs to be a network
The graph $D$ in \Cref{def:tail} and \Cref{def:head} is not necessarily a network, as it may be cyclic or contain parallel edges. If that is the case, then we say the tail move or the head move is \textit{invalid}. All other moves, where $D$ is a network, are called \textit{valid} moves. 

Note that head moves and tail moves do not change the number of leaves, tree nodes, or reticulations of a network. Head moves and tail moves together are called \emph{rSPR} moves. In fact, most commonly used horizontal rearrangement moves for phylogenetic networks are a type of $\rSPR$ move. Thus, if we consider move proposals for horizontal rearrangement moves, then we can restrict our attention to $\rSPR$ moves.

Note that two rearrangement moves that lead to the same leaf-labelled topology may result in different networks when all internal labels (i.e., non-leaf labels) are considered (Figure~\ref{fig:InternalLabelsExample}).
\begin{figure}[ht]
    \centering
    \includegraphics[width=.75\textwidth]{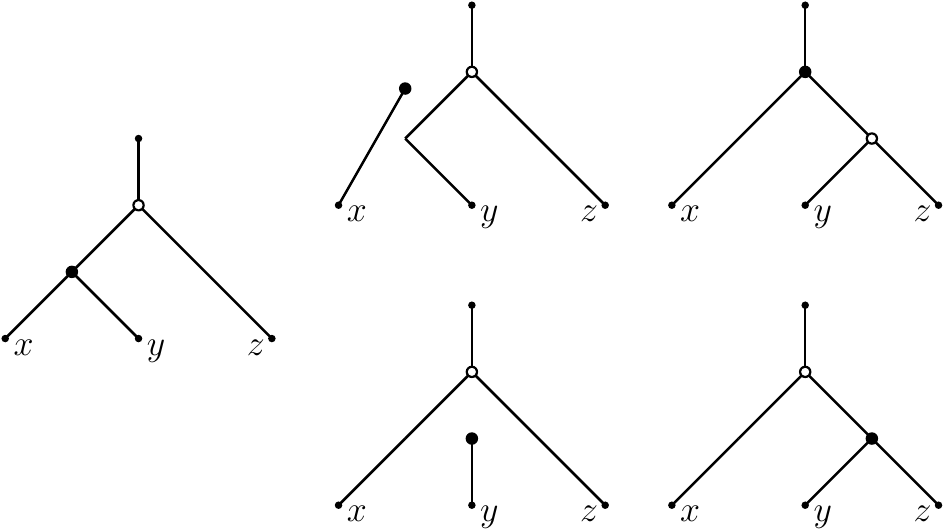}
    \caption{Two tail moves (defined in Section~\ref{sec:rearr}) applied to a network with internal labels. The figure explicitly shows the pruning step, resulting in the networks in the middle, and the reattachment step, resulting in the two networks on the right. The moves applied to the left network result in two networks that are leaf-isomorphic, but not isomorphic with respect to all internal labels. The labels of the internal nodes are represented by the node shapes, the labels of the leaves are shown as letters.}
    \label{fig:InternalLabelsExample}
\end{figure}

\paragraph{Vertical Moves}
Although several types of vertical moves have been defined \cite{huber2016spaces,wen2016bayesian,bordewich2017lost,gambette2017rearrangement,janssen2019rearrangement}, we restrict our attention to one simple version which simply adds or removes an edge~\cite{bordewich2017lost}.

\begin{definition}\label{def:vplus}
Let $N$ be a network containing the distinct edges $(u,v)$ and $(u',v')$. Let $D$ be the digraph obtained from $N$ by replacing $(u,v)$ and $(u',v')$ with edges $(u,w)$, $(w,v)$, $(u',w')$, $(w',v')$, and $(w,w')$. If $D$ is a network, then we say $D$ is the result of the valid $\VP$ move $\VP((u,v),(u',v'))$, otherwise, we say that $\VP((u,v),(u',v'))$ is invalid. For fully labelled networks, the labels $l(w)$ and $l(w')$ of the new nodes $w$ and $w'$ must also be defined, which we denote $\VP((u,v),(u',v'),l(w),l(w'))$.
\end{definition}

\begin{definition}\label{def:vmin}
Let $N$ be a network containing the distinct edges $(u,w)$, $(w,v)$, $(u',w')$, $(w',v')$, and $(w,w')$. Let $D$ be the digraph obtained from $N$ by replacing the edges $(u,w)$, $(w,v)$, $(u',w')$, $(w',v')$, and $(w,w')$ with the edges $(u,v)$ and $(u',v')$, and removing the nodes $w$ and $W'$. If $D$ is a network, then we say $D$ is the result of the $\VM$ move $\VM(w,w')$ and $\VM(w,w')$ is valid, otherwise, we say that $\VM(w,w')$ is invalid.
\end{definition}

Note that the result $D$ of this process can be invalid for additional reasons compared to horizontal moves. If $w$ is a reticulation node, then $D$ contains a node of indegree-$2$ and outdegree-$0$ which is not allowed in a network. Similarly, if $w'$ is a tree node, then $D$ contains a node of indegree-$0$ and outdegree-$2$.

\subsubsection{Spaces of Networks}
A set of networks together with rearrangement moves forms a space of networks. Such a space can be viewed as a graph where the set of nodes is a set of networks, and there is an edge from one network to another if there is a rearrangement move turning the one into the other. To define these spaces more rigorously, we first set some notation for certain sets of networks.

\begin{definition}
Let $n\geq 1$ and~$k\geq 0$. The set of all networks with $n$ leaves and $k$ reticulations is denoted $\N(n,k)$, and the set of all networks with $n$ leaves and \textit{at most} $k\leq k_{\mathrm{max}}$ reticulations is denoted $\N(n,\leq k_{\mathrm{max}})$. The label set $X=X^l=\{x^l_1,\ldots,x^l_n\}$, which labels only the leaves, of all these networks is fixed.

Moreover, $\Ni(n,k)$ (resp. $\Ni(n,\leq k_{\mathrm{max}})$) is the set of fully labelled networks with $n$ leaves and $k$ (resp. at most $k$) reticulations. Again, the label sets are fixed for all these networks to $X^l=\{x^l_1,\ldots,x^l_n\}$, $X^r=\{x^r_1,\ldots,x^r_k\}$, and $X^t=\{x^t_1,\ldots,x^t_{n+k-1}\}$. 
\end{definition}

%DEF: Space of phylogenetic networks
\begin{definition}
The \emph{space of (phylogenetic) networks} defined by the directed rearrangement moves of type $M$ with $n$ leaves and $k$ reticulations, is the graph $\N[M](n,k)$ (resp. $\N[M](n,\leq k_{\mathrm{max}})$) whose set of nodes is the set of networks $\N(n,k)$ (resp. $\N(n,\leq k_{\mathrm{max}})$) with $n$ leaves and $k$ (resp. at most $k$) reticulations, and there is an edge between two networks if there is an $M$-move that transforms the one into the other.
The spaces $\Ni[M](n,k)$ and $\Ni[M](n,\leq k_{\mathrm{max}})$ are defined similarly.
\end{definition}

%DEF: Connected space
The \emph{$k$-th tier} \cite{francis2017bounds,huber2016transforming} of network space consists of all networks with $k$ reticulations. We say the $k$-th tier is connected by a (directed) horizontal move type $M$ if all $\N[M](n,k)$ are connected, and the space of networks is connected by a horizontal move type $M$ if $\N[M](n,k)$ is connected for all $n,k\geq 0$. This implies that using such a horizontal move type in conjunction with any vertical move type for which there is at least one move between each subsequent pair of tiers, any network can be reached from any other network.

It has been proven that most $\N[M](n,k)$ and $\Ni[M](n,k)$ are connected under $\rSPR$ moves with a relatively small diameter (i.e., largest distance between any two nodes) \cite{janssen2018exploring, thesis_janssen}. Moreover, this still holds for more restricted moves such as head moves, tail moves, and rNNI moves\cite{janssen2019rearrangement, thesis_janssen}, and also when $\rSPR$ moves are applied to restricted spaces such as those formed by certain classes of networks \cite{klawitter2020spaces}.

\section{Symmetries and Proposal Ratios}\label{sec:ProposalRatios}
In this section, we apply the theory from Section~\ref{sec:MCMC} to leaf-labelled networks as a quotient of fully labelled networks. We show how we can sample from a set of leaf-labelled networks by using a correction factor on the acceptance ratio for move proposals on fully labelled networks. It will follow directly from the theory above that this ratio does not depend on whether the space of fully labelled networks is connected under a chosen rearrangement move; only the space of leaf-labelled networks needs to be connected. Perhaps slightly more surprisingly, the acceptance ratio is also independent of the number of labels used for the internal nodes. 

Computing the acceptance ratio does require the number of fully labelled networks that represent one leaf-labelled network $N$. For this, we will study the symmetry of phylogenetic networks, or, more specifically, the number of automorphisms $|\Aut(N)|$ of $N$.

\subsection{Proposal Ratios}
To sample from sets of leaf-labelled networks (i.e. $\N(n,\leq k_{\mathrm{max}})$ and $\N(n, k)$), we use a proposal Markov chain based on the rearrangement moves from Section~\ref{sec:rearr}, applied to fully labelled networks. Move proposals are generated by selecting a move type at random independent of the current network (if multiple are used), and then generating a move proposal for that move type and the current network, by picking some nodes, edges, and labels uniformly at random --- for $\VP$ moves the new labels need to be selected as well. For example, we may pick an $\rSPR$, a $\VM$, or a $\VP$ move type with equal probability. Say we picked $\VM$, then the next step is to generate a move proposal of that type for the current network, which, in this case of $\VM$, is simply a random edge in the network. This is quite similar to the uniquely reversible moves from PhyloNet \cite{wen2016bayesian} and SpeciesNetwork \cite{zhang2018bayesian}, although these studies seem to assume that this reversibility extends to leaf-labelled networks.

More formally, we generate a move type $M$ with a fixed probability $\PP(M)$. Thus, a move proposal $m\in M$ on $N$ is generated with probability $g(m)=\PP(M) g(M,n,k)$. Note that the proposal probability $g(M,n,k)$ does not depend on the topology of $N$ but only on the move type $M$, and the network parameters $n$ (the number of leaves) and $k$ (the number of reticulations). This is because these already determine the number of nodes and edges in the network, and we only pick nodes and edges uniformly at random to generate a move proposal, which in turn requires some number of nodes and edges dependent on the move type. This independence of the topology together with reversibility of move proposals makes it possible to calculate the ratios of proposal probabilities needed to use Metropolis-Hastings sampling. 

As mentioned in Section~\ref{sec:MCMC}, to sample leaf-labelled networks, we now need to calculate the Hastings ratio for the quotient Markov chain induced by the group action that permutes all internal labels. Indeed, for that group action, each equivalence class represents a maximal set of labelled networks that are isomorphic as leaf-labelled networks. 

\subsubsection{Proposal Probabilities Worked Out}\label{sec:proposal_probabilities_worked_out}

As mentioned before, $\rSPR$ moves can be represented by a tuple $(e_m, v_m, e_t)$, which represents the move $v_m\move{e_m} e_t$. If we sample the two edges uniformly at random from the set of edges and then pick an endpoint of $e_m$ for $v_m$ uniformly at random, all from a network $N$ with $n$ leaves and $k$ reticulations, the proposal probability works out to 
\[g(\rSPR, n,k) = \frac{1}{2(2n+3k-1)^2}.\]

A tail move can be represented by a pair of edges $(e_m, e_t)$, representing the move $v_m\move{e_m} e_t$ where $v_m$ is the tail end of $e_m$. Again, by uniformly at random picking two edges from the set of all edges, we get a proposal probability for tail moves. This works out to 
\[g(\Tail, n,k) = \frac{1}{(2n+3k-1)^2}.\] 

Similarly, for head moves, a pair of edges $(e_m, e_t)$ represents the move $v_m\move{e_m} e_t$ where $v_m$ is the head end of $e_m$. This leads to the same proposal probability 
\[g(\Head, n,k) = \frac{1}{(2n+3k-1)^2}.\]

A $\VP$ move proposal can be represented by a tuple $(e_1, e_2, l_1, l_2)$, which represents the move $\VP(e_1,e_2,l_1,l_2)$ (Definition~\ref{def:vplus}). Here, $e_1$ and $e_2$ are edges of the current network, and $l_1$ and $l_2$ are labels chosen from the full label set $L$. In this case, we generate the edges independently by picking each uniformly at random from the edges of the network, but the labels are generated dependently, by picking one after the other uniformly at random from the labels not used by the network. 
Denoting the leaf labels with $K$, and all other labels with $L$, there are $|L|-|V(N)|+|K|$ such unused labels.  
Hence, the proposal probability for a $\VP$ move becomes 
\[g(\VP, n,k) = \frac{1}{(2n+3k-1)^2 (|L|-|V(N)|+|K|) (|L|-|V(N)|+|K|-1)}.\]

A $\VM$ move proposal can be encoded by a single edge $e$ (the removed edge), representing the move $\VM(e)$ (Definition~\ref{def:vmin}). By again selecting the edge uniformly at random, we get a proposal probability 
\[g(\VM, n,k) = \frac{1}{2n+3k-1}.\]

Note that an $\rSPR$, $\Tail$, or $\Head$ move does not change the number of leaves or the number of reticulations, and the reverse of such a move is again an $\rSPR$, $\Tail$, or $\Head$ move respectively. Hence, the probability of a proposal and its reverse are the same for these move types. This means that the proposal ratio part of the Metropolis-Hastings acceptance ratio for such a Markov chain on fully labelled networks will always be equal to 1, and the only remaining part is the ratio of the target distribution for the two states. 

This is not the case for vertical moves, where the reverse proposal of a $\VP$ move is a $\VM$ move, as the number of reticulations changes after such moves. Hence, for a $\VP$ move on a network with $n$ leaves and $k$ reticulations, the ratio of proposal probabilities in the acceptance ratio is 
\[ \frac{g(\VM, n,k+1)}{g(\VP, n,k)} = \frac{(2n+3k-1)^2 (|L|-|V(N)|+|K|) (|L|-|V(N)|+|K|-1)}{2n + 3k + 2}.\]

Similarly, for a $\VM$ move on a network with $n$ leaves and $k$ reticulations, the ratio of proposal probabilities in the acceptance ratio is 
\[ \frac{g(\VP, n,k-1)}{g(\VM, n,k)} = \frac{2n + 3k - 1}{(2n+3k-4)^2 (|L|-|V(N)|+|K|) (|L|-|V(N)|+|K|-1)}.\]

\subsubsection{The Hastings Ratio for the Quotient Markov Chain}
To calculate the Metropolis-Hastings ratio for leaf-labelled networks we must add a correction to the ratio of proposal probabilities $g(M,n,k)/g(M',n',k')$ (Theorem~\ref{the:quotient_MH}). This correction is simply a correction for the number of representations of each leaf-labelled network by fully labelled networks (Section~\ref{sec:correction_representations}). Following Equation~\ref{eq:correction_representation}, this is a function of the cardinality of the label set, the number of labels used in the network, and the cardinality of the automorphism set of the network. Combining all this, we get the following acceptance functions for the different move types.

\begin{theorem}\label{the:MH_for_leaf_labelled_networks}
    Let $\pi$ be a distribution on a space of leaf-labelled networks with non-zero probability for each network. Let $N$ be a leaf-labelled network with $n$ leaves and $k$ reticulations, and $\widetilde{m}:N \to N'$ a rearrangement move generated as outlined in Section~\ref{sec:proposal_probabilities_worked_out}. Provided the quotient proposal Markov chain on the leaf-labelled networks is irreducible and aperiodic, a Metropolis-Hasting Markov Chain with stationary distribution $\pi$ can be realised with the acceptance probability function
    \[
        A(\widetilde{m}) = \min\left(1, \frac{\pi(N')}{\pi(N)}\frac{|\Aut(N')|}{|\Aut(N)|}\right)
    \]
    for moves of type $\rSPR$, $\Tail$ and $\Head$, the acceptance function
    \[
        A(\widetilde{m}) = \min\left(1, \frac{\pi(N')}{\pi(N)}\frac{|\Aut(N')|}{|\Aut(N)|} \frac{(2n + 3k - 1)^2}{2n + 3k + 2} \right)
    \]
    for $\VP$ moves, and the acceptance function 
    \[
        A(\widetilde{m}) = \min\left(1, \frac{\pi(N')}{\pi(N)}\frac{|\Aut(N')|}{|\Aut(N)|} \frac{2n + 3k -1}{(2n + 3k - 4)^2}\right)
    \]
    for $\VM$ moves.
\end{theorem}
\begin{proof}
    Let $\dot{N}$ be a fully labelled representative of $N$. Translating the notation to that of Theorem~\ref{the:quotient_MH}, a ``leaf-labelled'' move $\widetilde{m}$ is generated by generating a move $m:\dot{N} \to \dot{N}'$ with probability $g(m) = g(M,n,k)$ and then selecting its equivalence class $\widetilde{m}: N \to N'$. We will use the $\Rep$ notation for the number of representatives in an equivalence class (instead of $|U|$). To prove the theorem, we work out $\Rep(N)$, $\Rep{N'}$, $g(m)$, and $g(m^{-1})$ in the acceptance function from Theorem~\ref{the:quotient_MH} for the different types of moves. Recall the acceptance function
    \[
        A(\widetilde{m}) = \min\left(1, \frac{\pi(N')}{\pi(N)}\frac{\Rep(N)}{\Rep(N')}\frac{g(m^{-1})}{g(m)} \right),
    \]
    and the number of representations (Equation~\ref{eq:correction_representation}):
    \[ 
        \Rep(N) = \frac{|L|!}{(|L|-|V(N)|+|K|)!} \frac{1}{|\Aut(N)|},
    \]

    As we already noted in Section~\ref{sec:proposal_probabilities_worked_out}, for moves of type $\rSPR$, $\Tail$ and $\Head$, the proposal probabilities for a move and its reverse are equal. Hence, for these types of moves, the quotient $\frac{g(m^{-1})}{g(m)}$ cancels. Moreover, the number of nodes ($|V(N)|$ and $|V(N')|$) are equal in both networks, too. Hence, in the quotient correcting for the number of representations, only the number of automorphisms are left, giving the acceptance probability we set out to prove.

    For vertical moves, the calculation is slightly more involved, as the quotients do not cancel out as neatly on themselves, but only once combined. Hence, we simply work out the equation. We start with $\widetilde{m}$ a $\VP$ move, where we write $T := |L|-|V(N)|+|K|$ for the number of unused labels in $N$:
    \begin{align*}
        A(\widetilde{m}) 
            &= \min\left(1, \frac{\pi(N')}{\pi(N)}\frac{\Rep(N)}{\Rep(N')}\frac{g(m^{-1})}{g(m)} \right)\\
            % &= \min\left(1, \frac{\pi(N')}{\pi(N)} \frac{|\Aut(N')|}{|\Aut(N)|} \frac{\frac{|L|!}{(|L|-|V(N)|+|K|)!}}{\frac{|L|!}{(|L|-|V(N')|+|K|)!}} \frac{(2n+3k-1)^2 (|L|-|V(N)|+|K|) (|L|-|V(N)|+|K|-1)}{2n + 3k + 2} \right)\\
            &= \min\left(1, \frac{\pi(N')}{\pi(N)} \frac{|\Aut(N')|}{|\Aut(N)|} \frac{\frac{|L|!}{T!}}{\frac{|L|!}{(|L|-|V(N')|+|K|)!}} \frac{(2n+3k-1)^2 T (T-1)}{2n + 3k + 2} \right)\\
            &= \min\left(1, \frac{\pi(N')}{\pi(N)} \frac{|\Aut(N')|}{|\Aut(N)|} \frac{(T-2)!}{T!} \frac{(2n+3k-1)^2 T (T-1)}{2n + 3k + 2} \right)\\
            &= \min\left(1, \frac{\pi(N')}{\pi(N)} \frac{|\Aut(N')|}{|\Aut(N)|} \frac{(T-2)!(T)(T-1)}{T!} \frac{(2n+3k-1)^2}{2n + 3k + 2} \right)\\
            &= \min\left(1, \frac{\pi(N')}{\pi(N)} \frac{|\Aut(N')|}{|\Aut(N)|} \frac{(2n+3k-1)^2}{2n + 3k + 2} \right)\\
    \end{align*}

    The number of labels cancel in a similar way for $\VM$ moves, which gives the acceptance probability we set out to prove.
\end{proof}

Interestingly, the probability is independent of the total number of labels used for all move types. Therefore, in practice, the full label set can be ignored. Note that we have not used any particular properties of the moves for this property to emerge: it was enough to look at the number of permutations of the labels in the network, and the number of labels added or removed by that move type. In fact, the moves do not strictly need to be generated like in Section~\ref{sec:proposal_probabilities_worked_out} for these acceptance probabilities to hold; the proof works as long as the probability for each move proposal is the same given a move type and a number of leaves and reticulations, and additional labels are drawn from the unused labels one after the other uniformly at random. For example, an $\rSPR$ move could equally well be generated by sampling two edges and a node uniformly at random and independently from the edges and nodes in the network. Contrast this ``inefficient'' method with the one from Section~\ref{sec:proposal_probabilities_worked_out} where we pick a moving endpoint from the moving edge. For the inefficient method, most move proposals are invalid, as the moving endpoint must per chance be an endpoint of the moving edge.

The theorem requires the proposal Markov chain to be aperiodic. Note that this is almost guaranteed for these types of chains, as invalid proposals are possible for nearly all networks and move types. More precisely, for any type of move in $\{\rSPR, \Tail, \Head, \VP, \VM\}$ and any space $\N(n,k)$ or $\N(n, \leq k_{\mathrm{max}})$ with $n\geq 2$ and $k\geq 0$, there always exists an invalid proposal for some network in the space. This implies a non-zero probability of staying at the same network, which in combination with irreducibility is enough to prove aperiodicity. 

\subsection{Counting Symmetries}
To sample (uniformly) from a class of networks whose networks all have a trivial automorphism group, we may simply put the condition of a network being within the class in the stationary distribution: by setting $\pi(N)=1$ when $N$ is in the class, and $\pi(N)=0$ otherwise. To do this efficiently, we need an efficient algorithm to determine whether a network belongs to the class. If we consider a class where all networks have $|\Aut(N)|=1$, calculating the Hastings ratio becomes almost trivial, as we do not need to compute automorphisms anymore. 

As an example, the class of tree-child networks can easily be sampled using this method: tree-child networks have a trivial automorphism group \cite{mcdiarmid2015counting,Klawitter2018SNPRneigh}, it can be checked in linear time whether a given network is tree-child \cite[e.g.]{huber2024orienting}, and the space of tree-child networks is connected if the reticulation number is small enough \cite{bordewich2017lost,klawitter2020spaces}. 

For the larger class of tree-based networks, we know the condition of having a trivial automorphism group is not met, but MCMC sampling is still possible if one is prepared to compute the number of automorphisms of a network to calculate the Metropolis-Hastings ratio for each proposal. Indeed, for tree-based networks, it is known that the tiers are connected under $\rSPR$ moves \cite{ERDOS2021205}, and it can be checked in linear time whether a network is tree-based \cite{zhang2016tree, hayamizu2021structure}. In the next section, we give an algorithm for computing the number of automorphisms of a leaf-labelled network.

The class of orchard networks is contained in the class of tree-based networks and it encloses the class of tree-child networks. For this class most conditions are known to be met as well: it can be checked in linear time whether a network is orchard \cite{erdHos2019class,janssen2021cherry}, and the space of orchard networks is connected under $\rSPR$ moves \cite{van2022orchard}. After providing an algorithm for computing the number of automorphisms of a leaf-labelled network, we will prove that orchard networks have a trivial automorphism group (Theorem~\ref{the:OrchardTrivAut}), just like tree-child networks. 

\subsubsection{Counting Non-Trivial Automorphisms}\label{sec:counting_automorphisms}
Finding the number of automorphisms of a graph is generally quite hard. More specifically, it is GI-hard because it is polynomial time equivalent to \textsc{Graph Isomorphism} (GI), the problem of checking whether two graphs are isomorphic \cite{mathon1979note,beals1999finding}.

The results of \cite{mathon1979note,beals1999finding} also imply that when restricted to a subset of graphs, calculating the number of automorphisms is in the same complexity class as GI restricted to that subset. Hence, calculating the number of automorphisms of a phylogenetic network is GI-hard in general \cite{cardona2014comparison}, but polynomial time solvable for binary networks \cite{luks1982isomorphism,kobler2012graph,mena2012ternary}. The algorithm or reduction that can be used for this purpose is the following (adapted from the $\textsc{ACOUNT} \propto_P \textsc{ISO}$ reduction in \cite{mathon1979note}). In the algorithm and related proofs, we use the notation $\id_V:V\to V$ for the identity map on $V$ (i.e., $\id_V(v)=v$ for all $v\in V$) and we say a map $\phi:X\to Y$ \textit{extends} a map $\psi: X' \to Y$ if $X'\subseteq X$ and $\phi(x)=\psi(x)$ for all $x\in X'$.

\vspace{\baselineskip}
\begin{algorithm}[ht]
 \KwData{A network $N$.}
 \KwResult{The number of automorphisms of $N$}
Set $c=1$\;
Set $V_f=V(N)$\;
\While{$V_f\neq L(N)$}{
    Let $v\in V_f\setminus L(N)$ be arbitrary\; 
    Set $c_v=1$\;
    \For{$w\in V(N)\setminus V_f$}{
        \If{there is an automorphism $\alpha$ of $N$ that extends $\id_{V_f\setminus\{v\}}(x)$ such that $\alpha(v)=w$}{\label{line:isom_check}
            Set $c_v=c_v+1$\;
        }
    }
    Set $c=c\cdot c_v$\;
    Remove $v$ from $V_f$\;
}
\Return $c$\;
\caption{\textsc{AutomorphismCount}$(N)$}\label{alg:AUTOM}
\end{algorithm}
\vspace{\baselineskip}

The correctness of this algorithm is a direct consequence of the following result by Mathon, of which we reproduce the proof in our notation. The result itself is a rather direct consequence of the orbit-stabilizer theorem.

\begin{lemma}[Mathon Reduction~$\mathrm{ACOUNT} \propto \mathrm{ISO}$]\label{lem:aut_count_step}
    Let $G$ be a graph $G$ and $S\subseteq V(G)$ a subset of the vertices of $G$ and write $\Aut(G_S):=\{\alpha\in \Aut(G): \alpha \mathrm{~extends~} \id_S \}$ for the automorphisms of $G$ that fix $S$. If $s\in S$ and $c:=|\{v\in V(G): \exists \alpha \in \Aut(G_{S\setminus \{s\}}) \mathrm{~s.t.~} \alpha(s)=v \}|$ is the number of nodes $v\in V(G)$ such that there is an automorphism $\alpha$ that extends $\id_{S\setminus \{s\}}$ and maps $s$ to $v$, then $|\Aut(G_{S\setminus s})| = c \cdot |\Aut(G_S)|$.
\end{lemma}
\begin{proof}
    First note that $\Aut(G_S)$ is a subgroup of $\Aut(G_{S\setminus\{s\}})$. Let $X=\{v\in V(G): \exists \alpha \in \Aut(G_{S\setminus \{s\}}) \mathrm{~s.t.~} \alpha(s)=v \}$ and thus $c=|X|$ as defined in the lemma statement and, for each $x\in X$, let $\phi_x$ be an automorphism in $\Aut(G_{S\setminus\{s\}})$ which maps $s$ onto $x$. Then every element $\tau \in \Aut(G_{S})$ is the product of a unique $\phi_x$ and a unique $\psi \in \Aut(G_{S})$. Therefore, we have that $|\Aut(G_{S\setminus\{s\}})| = c \cdot \Aut(G_{S})$.
\end{proof}

\begin{lemma}\label{lem:algAUTOM}
\Cref{alg:AUTOM} finds the number of automorphisms of a given network in~$O(n^2f(n))$ time, where~$n$ is the number of nodes in the network, and~$O(f(n))$ is the time complexity of the `binary network isomorphism problem'.
\end{lemma}
\begin{proof}
    We will first prove correctness of the algorithm. Let the nodes in $V(N)\setminus L(N)$ be indexed $v_1, \ldots, v_n$ in the order they are removed from $V_f$ in the algorithm.
    Note that in iteration $i$ of the main loop, $v=v_i$ and $V_f=\{v_i,\ldots, v_n\}\cup L(N)$. 
    Let $c_i$ be the value of $c_v$ at the end of iteration $i$ of the main loop, then $c_i$ is the number of nodes $w \in V(N)\setminus \{v_1, \ldots, v_i\}$ for which there is an automorphism $\alpha$ that extends $\id_{\{v_{i+1},\ldots, v_n\}}$ and maps $\alpha(v_i)=w$.
    Note that $c_v$ is set to $1$ initially to account for the automorphism $id_{V(N)}$. 
    In other words, $c_i=|\{w\in V(N): \exists \alpha \in \Aut(N_{\{v_{i+1},\ldots, v_n\}}) \mathrm{~s.t.~} \alpha(v_i)=w \}|$.
    Finally, note that at the end of the algorithm, $c = \prod_{i=1}^n c_i$, which by Lemma~\ref{lem:aut_count_step} gives
    \[c = \prod_{i=1}^n c_i = \prod_{i=1}^n \frac{|\Aut(N_{\{v_{i+1},\ldots, v_n\}})|}{|\Aut(N_{\{v_i,\ldots, v_n\}})|} = \frac{|\Aut(N_{\emptyset})|}{|\Aut(N_{\{v_1,\ldots, v_n\})}|} = |\Aut(N)|, \]
    where $\{v_{n+1},\ldots, v_n\} = \emptyset$, and $|\Aut(N_{\{v_1,\ldots, v_n\}})|=1$ because there is exactly one automorphism that fixes all the nodes, namely the identity map.

    The running time of $O(n^2 f(n))$ is computed by counting the number of times network isomorphism has to be computed, as that is the only non-constant step in the algorithm. 
    The number of times is exactly the total number of iterations of the inner loop of the algorithm, which is the number of iterations of the outer loop times the maximal number of iterations of the internal loop per iteration of the outer loop.
    The outer loop has fewer than $n=|V(G)|$ iterations, as the algorithm starts with $V_f=V(N)$ and ends with $V_f=L(N)$ while removing exactly one element of $V_f$ per iteration of the loop. 
    The inner loop has $|V(N)\setminus V_f| \leq | V(N) | =n$ iterations per iteration of the outer loop. Hence, the total number of iterations of the inner loop is bounded by $n^2$. 
    Finally, because the algorithm checks for an isomorphism exactly once per iteration with a running time of $f(n)$, the running time of \Cref{alg:AUTOM} is $O(n^2 f(n))$.
\end{proof}

\Cref{alg:AUTOM} has been implemented as the function \href{https://github.com/RemieJanssen/PhyloX/blob/5dccfd24e0376f9edd7c0695bf4c701a6c8ead52/src/phylox/isomorphism/base.py\#L141C18-L141C19}{\texttt{count\_automorphisms}} in the Python package PhyloX \cite{janssen2024phylox}.
For the check in line~\ref{line:isom_check} of the algorithm, the best known running time in general is $O(n^{10})$ \cite{mena2012ternary}. Our implementation uses the VF2 algorithm from \cite{cordella2001improved} through the NetworkX~\cite{networkx} function \href{https://networkx.org/documentation/stable/reference/algorithms/generated/networkx.algorithms.isomorphism.is_isomorphic.html}{\texttt{is\_isomorphic}}, which may not guarantee a polynomial running time, but runs fast in practice.

Note that there may well be better algorithms for this task, but no algorithms tailored specifically to phylogenetic networks have been published to date. It is quite easy to improve this algorithm heuristically---for example by only considering $w$ of the same node type as $v$---but it is unclear if this would improve the theoretical running time. It would be interesting to see whether a quadratic or linear time algorithm is attainable, and whether the problem becomes easier for classes of networks such as tree-based networks. Here, we will first improve on the previous algorithm heuristically using $\mu$-representations of networks, and then we will show that for some classes of networks, the problem becomes irrelevant, because these networks cannot have non-trivial symmetry.

% \subsubsection{Using \texorpdfstring{$\mu$}{mu}-Representations to Speed Up \textsc{AutomorphismCount}}
\subsubsection{Using \texorpdfstring{$\mu$}{mu}-Representations to Speed Up
{\normalfont\texorpdfstring{\textsc{AutomorphismCount}}{AutomorphismCount}}}
A heuristic speedup of \textsc{AutomorphismCount} can be achieved by considering some properties of nodes in the network that can be computed quickly in advance. The simplest example, as noted earlier, would be to first check if a node is a tree node or a reticulation, as nodes of different types cannot map to each other in an automorphism. A property that holds much more information is the $\mu$-vector of each node. In this section, we show how to improve \textsc{AutomorphismCount} by using these $\mu$-vectors.

\paragraph{$\bm{\mu}$-Representation of Networks}

Let~$N$ be a network~$N$ on a leaf-set~$X = [n] = \{1,2,\ldots,n\}$ (simply give some ordering on~$X$).
The~\emph{$\mu$-vector}~$\mu(v)$ of a vertex~$v$ is the~$n$-tuple where the~$i$th element denotes the number of paths from~$v$ to the leaf~$i$.
We denote this by writing~$\mu(v) = (m_1(v), m_2(v), \ldots, m_n(v))$, where~$m_i(v)$ is the number of paths from~$v$ to the leaf~$i$.
The \emph{$\mu$-representation} of~$N$ is the multi-set of~$\mu$-vectors for every vertex in~$N$.

The $\mu$-representation of networks was introduced by Cardona et al. to study unique characterisation results for tree-child and time-consistent tree-sibling networks, but also to give distance metric on the respective classes~\cite{cardona2008distance,cardona2009comparison}.
It has also been shown in recent literature that stack-free orchard networks are uniquely characterised by their~$\mu$-representations (here, the~$\mu$-representations were called ancestral profiles; it should be noted that these two terms encode the same amount of information, and are therefore equivalent.)~\cite{erdHos2019class,bai2021defining}.

\begin{observation}\label{obs:AutSameMu}
Let~$N$ be a network and let~$f$ be an automorphism on~$N$, then~$f$ maps every vertex to a vertex of the same type (root, tree vertex, reticulation, or leaf) and we have~$\mu(v) = \mu(f(v))$ for every vertex~$v$ in~$N$.
\end{observation}
\begin{proof}
Since~$f$ is an automorphism, we have that~$(u,v)$ is an edge of~$N$ if and only if~$(f(u),f(v))$ is an edge of~$N$.
For the in-degrees and out-degrees of vertices to coincide, we must have that vertices of a certain type are mapped to vertices of the same type.
To prove the second statement, observe that paths are mapped to distinct paths under~$f$.
Therefore, the number of paths to a leaf is preserved under~$f$.
\end{proof}

The following lemma makes it clear that the $\mu$-vectors of a network can be computed in one bottom-up pass of the network. Hence, pre-computing these $\mu$-vectors only adds a quadratic component ($O(|V(G)|\cdot |X|)$) to the running time of the automorphism computation.

\begin{lemma}[Lemma 4b of \cite{cardona2009comparison}]\label{lem:MuChildren}
Let~$N$ be a network, and let~$v$ be a vertex of~$N$ with children~$c_1,\ldots,c_k$, then~$\mu(v) = \sum_{i\in [k]} \mu(c_i)$.
\end{lemma}

\paragraph{Algorithm and Its Proof} We first introduce a few simple notational conveniences for the definition of the algorithm, and for the proof of its validity.
Let $N = (V,E)$ be a binary phylogenetic network with $\mu$-representation $\mu(N)$ and leaf set $X$, then we denote by $t(v)$ the type of node $v$ (tree-node or reticulation).

\begin{algorithm}[ht]
 \KwData{A network $N$ and its $\mu$-representation $\mu(N)$.}
 \KwResult{The number of automorphisms of $N$}
    Set $c = 1$\;
    Set $V_f = V(N)$\;
    \While{$V_f \neq L(N)$}{
        Let $v \in V_f\setminus L(N)$ be arbitrary\;
        Set $c_v = 1$\;
        Set $S = \{v_m \in V(N)\setminus V_f: \mu(v_m) = \mu(v), t(v_m) = t(v)\}$\;
        \ForEach{$w \in S$}{
            \If{\textit{there is an automorphism $\alpha$ of $\N$ that extends $\id_{V_f \setminus \{v\}}$ such that $\alpha(v) = w$}}{
                Set $c_v = c_v + 1$\;
            }
        }
        Set $c = c_v \cdot c$\;
        Set $V_f = V_f \setminus \{v\}$\;
    }
    \Return $c$\;
\caption{\textsc{AutomorphismCountWithMu}$(N, \mu(N))$}\label{alg:automorphism_count}
\end{algorithm}

\begin{lemma}
Given a network $N$ and its $\mu$-representation $\mu(N)$, \Cref{alg:automorphism_count} computes $|\Aut(N)|$ in $O(kn\cdot f(n))$ time, where $n$ is the number of nodes in the network, $k$ is the maximal size of the group of nodes with the same $\mu$-vector, and $O(f(n))$ is the time complexity of the `binary network isomorphism problem'. 
\end{lemma}
\begin{proof}
    For correctness, note that the only difference between \Cref{alg:AUTOM} and \Cref{alg:automorphism_count} is the set over which the inner loop iterates. The set $S=\{v_m \in V(N)\setminus V_f: \mu(v_m) = \mu(v), t(v_m) = t(v)\}$ of the second algorithm is a subset of $V(N)\setminus V_f$ which is used in the first algorithm. Hence, to show correctness, we must simply show that the value of $c_v$ at the end of an iteration of the outer loop is the same for both algorithms.
    
    To this end, we observe that the if condition in the inner loop can only hold if $w$ and $v$ have the same automorphism invariant properties. In other words, suppose that $p:V(N) \to P$ is a node-property such that for every automorphism $\alpha$ of $N$ and every $v\in V(N)$ we have $p(v)=p(\alpha(v))$, then the if condition can only be true if $p(v)=p(w)$, and it is safe to loop over $\{v_m \in V(N)\setminus V_f: p(v_m) = p(v)\}$ in the inner loop.

    What rests is to show that $\mu$ and $t$ are such properties. The fact that $\mu$ is an automorphism invariant property is proven in \Cref{obs:AutSameMu}. For $t$, we observe that incidence relations are preserved under an automorphism, and thus also the in-degree and out-degree of each node. As the node type $t(v)$ of a node $v$ is fully determined by its in-degree and out-degree, the node type $t$ is also an automorphism invariant property.
    
    For the running time, we use a similar calculation as in the proof of \Cref{lem:algAUTOM}. The number of iterations of the outer loop is still bounded by $n=|V(N)|$, and the number of iterations in the inner loop is bounded by the maximal size of $S$ which per definition is $k$: at most the maximal size of the group of nodes with the same $\mu$-vector. Hence, the worst case running time is $O(kn\cdot f(n))$.
\end{proof}

Both Algorithm~\ref{alg:AUTOM} and Algorithm~\ref{alg:automorphism_count} have been implemented in PhyloX \cite{janssen2024phylox}. We have tested the the running times of this implementation in Appendix~\ref{sec:Appendix_Running_Time_Mu}.

\subsubsection{Orchard Networks Are Without Symmetry}
In Lemma~5.1 of~\cite{mcdiarmid2015counting}, it is proven that tree-child networks have no non-trivial automorphisms that fix the leaves. This is not true for reticulation-visible and tree-based networks (\cite{klawitter2020spaces} Observation~4.34). The example that proves this is a simple 2-leaf network with a crown. Note that each tier-1 network is tree-child, so it has no non-trivial automorphisms. In this section, we will show that the class of orchard networks \cite{janssen2021cherry,erdHos2019class}---which contains the class of tree-child networks \cite{janssen2021cherry} and is contained in the class of tree-based networks \cite{van2022orchard}---is also without symmetry. To prove that orchard networks have no non-trivial automorphisms, we again use the theory of $\mu$-representations.

We say two vertices are \emph{clones} if they have the same~$\mu$-vectors.
If both vertices are tree vertices, we say they are \emph{tree clones}.
We call a network \emph{tree clone-free} if it contains no tree clones.
We show here that the number of automorphisms for a tree clone-free network is 1, namely the identity map.

\begin{proposition}[Theorem 25 of \cite{murakami2021thesis}]
\label{prop:TCFTrivAut}
Let~$N$ be a tree clone-free network, then~$|\Aut(N)| = 1$.
\end{proposition}
\begin{proof}
Let~$f$ be an automorphism from~$N$ to itself. We shall prove that~$f$ is the identity map.
As leaves are uniquely labelled, leaves are uniquely mapped to leaves of the same labels under~$f$.
By Observation~\ref{obs:AutSameMu}, tree vertices must be mapped to tree vertices of the same~$\mu$-vector. As~$N$ contains no tree clones, tree vertices are mapped to themselves under~$f$.
Therefore it remains to show that reticulations must also be mapped to itself under~$f$.

Suppose for a contradiction that this was not the case, i.e., that a reticulation is mapped to another reticulation under~$f$. 
Let~$r_1$ and~$r_2$ denote such two reticulations, where~$f(r_1) = r_2$.
In particular, choose~$r_1$ to be a highest such reticulation.
Because~$f$ is a bijection, we have that~$(u,v)$ is an edge of~$N$ if and only if~$(f(u),f(v))$ is an edge of~$N$.
This means that the parents of~$r_1$ must also be mapped to the parents of~$r_2$.
Since we chose~$r_1$ to be the highest such reticulation, the parents of~$r_1$ must both be tree vertices. 
Consequently, the parents of~$r_2$ must also both be tree vertices; otherwise, a tree vertex would be mapped to a reticulation, which cannot happen under an automorphism by Observation~\ref{obs:AutSameMu}.
But we know that tree vertices are mapped to themselves under~$f$.
It follows then that~$r_1$ and~$r_2$ have the same parents~$t_1$ and~$t_2$.
By Lemma~\ref{lem:MuChildren} then~$\mu(t_1) = \mu(r_1) + \mu(r_2) = \mu(t_2)$, which must mean that~$t_1$ and~$t_2$ are tree clones.
This contradicts the fact that~$N$ is tree clone-free.
Thus reticulations must be mapped to itself under~$f$, and therefore we have that~$f$ must be the identity map.
\end{proof}

It has been shown that orchard networks are tree clone-free (Lemma 4.4 of~\cite{bai2021defining}).
Therefore, Proposition~\ref{prop:TCFTrivAut} holds for orchard networks.

\begin{theorem}\label{the:OrchardTrivAut}
Let~$N$ be a leaf-labelled orchard network, then~$|\Aut(N)| = 1$.
\end{theorem}

As each tree-child network is orchard, and each tree is a tree-child network, we have the following corollaries. Both are well known, but the proof here is new.

\begin{corollary}
Let~$N$ be a leaf-labelled tree-child network, then~$|\Aut(N)| = 1$.
\end{corollary}

\begin{corollary}
Let~$N$ be a leaf-labelled tree, then~$|\Aut(N)| = 1$.
\end{corollary}

\section{Discussion}
In this paper, we have provided a Metropolis-Hastings sampling method for leaf-labelled phylogenetic networks. For this, we combined a variation on Metropolis-Hastings for quotient Markov chains with theory about spaces of phylogenetic networks. In the process, we have also studied the automorphism groups of leaf-labelled phylogenetic networks, and shown that for some classes of networks, these automorphism groups are trivial.

The main result regarding the Metropolis-Hastings sampling method is Theorem~\ref{the:MH_for_leaf_labelled_networks}, which provides Metropolis-Hastings acceptance probabilities for basic types of rearrangement moves: $\rSPR$ moves, $\Tail$ moves, $\Head$ moves, $\VP$ moves, and $\VM$ moves. These can easily be extended to other move types like $\rNNI$ moves, $\Tail_1$ moves, and $\Head_1$ moves by using the corresponding long-range moves (resp. $\rSPR$, $\Tail$, and $\Head$) and setting the probability to $0$ if the move is of distance longer than $1$. With a little more effort, one can also generate such moves with smaller probability of invalid moves, while retaining the property that each move proposal is equally probable. 

For example, for $\Tail_1$ moves the following might work: generate a move proposal by drawing an edge and a direction in (up-left, up-right, down-left, down-right) uniformly at random. For such a proposal, the chosen edge is the moving edge and the direction indicates where to move the tail: up is towards the root and down is towards the leaves, left and right are used to determine which of the upwards or downwards edges are chosen to move to. To ensure each move proposal is equally probable, right encodes an invalid move whenever there is only one upwards or downwards edge and up or down is chosen respectively. One still does have to check that such proposals uniquely encode moves to determine if this works.

\subsection{Correcting for Automorphisms with Branch Lengths}
The main point to take notice of in Theorem~\ref{the:MH_for_leaf_labelled_networks}, is that the acceptance probability needs to be corrected for the number of automorphisms of the networks before and after the move. This can be seen as a correction for the fact that sampling actually occurs in the space of fully labelled networks. The number of fully labelled networks corresponding to a leaf-labelled network is not equal for all networks, so networks with more representations (i.e. smaller automorphism group) will be over-represented in the sample. 

Extending this analysis to networks with branch lengths leads to an interesting situation. Suppose that the branch lengths are randomly sampled from a continuous distribution --- which they generally are --- the automorphism group of each network (including its branch lengths) becomes trivial with probability one. Indeed, with probability one, all edges of the network have different branch lengths. Hence, the probability that we can map an edge to any other edge in the network is zero. One could be inclined to conclude that, in practice, correcting for symmetry is no longer necessary when branch lengths are involved.

However, it should be noted that this overlooks the fact that leaf-labelled networks with non-trivial automorphism groups still exist. This is hidden in the argument above, because the parameter space (of branch lengths) of such 
a network $N$ is effectively $|\Aut(N)|$ times as small. All other things being equal, this makes the marginal probability for a leaf-labelled network topology with non-trivial automorphism group smaller than the probability of those with a trivial automorphism group. This is similar to BHV-like spaces, in which subspaces of different tree topologies may have different volumes \cite{gavryushkin2016space}.
This raises the question of how ``uniform'' prior distributions for Bayesian sampling methods should be defined. Should all network topologies be equally likely, or should each network including branch lengths be equally likely? 
Increased awareness of such questions fits into a broader trend of more careful design of Bayesian inference in the evolutionary setting (e.g., \cite{mendes2025validate}).

It would be interesting to study current implementations of sampling methods of posteriors in Bayesian methods for phylogenetic networks (e.g., PhyloNet \cite{wen2016bayesian} and BEAST 2.0 \cite{zhang2018bayesian}). Some version of this choice is already present as part of the implemented prior, be it explicit or implicit. The choice may also be relevant for the types of conclusions one can draw. For example, if one starts with a network topology and wants to check whether the method correctly reconstructs this network topology, it might make sense to start with a uniform prior on the network topologies.

One may question whether it is of any use to correct for automorphisms, as it seems unlikely that many networks even have non-trivial automorphisms. To make this more precise, it would be interesting to study the number of such networks in different classes of networks, and with varying number of leaves and reticulations. We conjecture that such networks vanish quickly with increasing number of leaves. 
Therefore, they would become irrelevant quickly when considering networks with a large number of leaves and a relatively small number of reticulations. 

It could be helpful to quantify this more precisely, especially because the converse appears to be true for increasing number of reticulations: we conjecture that the probability of non-trivial symmetry goes to 1 when considering networks with a fixed number of leaves and an increasing number of reticulations.

\subsection{Conclusions and Open Questions}
Although the correction for symmetries could be costly (w.r.t. computational time), there are good reasons to still include it in software. First, it gives the methods a sound mathematical basis, which is necessary even though, as we said, the effects may be small when networks with a small number of reticulations and a large number of leaves are considered.
Second, existing methods can only handle networks with a small number of leaves and a small number of reticulations. For these networks, the proportion of symmetric networks could be significantly higher. For example, one out of the eighteen networks in $\N(2,2)$ has non-trivial symmetry, and is thus undersampled if no correction is made (Figure~\ref{fig:undersampling_n2k2}). This could become relevant if Bayesian methods are used to build small networks, which are then combined using combinatorial methods such as binet, trinet or quarnet methods \cite{van2014trinets,oldman2016trilonet,huber2017reconstructing,huber2018quarnet,gross2020distinguishing}. Lastly, when these methods are improved so they can handle larger numbers of reticulations, the number of symmetries of the networks may also increase, making it more important to include a correction for this in the calculations.

To remedy the computational costs of adding a correction for symmetries, it would be interesting to see whether there are more efficient ways of computing $|\Aut(N)|$. As phylogenetic networks have a lot of structure that restricts the orbits of nodes, it may be possible to devise an algorithm that calculates $|\Aut(N)|$ much faster, leveraging this structure. For example, it would be interesting to see whether the $\mu$-vectors of the nodes can be used more effectively; or whether the problem becomes easier for classes of networks such as tree-based networks.

If it turns out these calculations of $|\Aut(N)|$ are still too costly, one could forego them completely, and resort to spaces of networks that have no symmetries, such as the space of tree-child networks \cite{mcdiarmid2015counting,Klawitter2018SNPRneigh} or the space of orchard networks (Theorem~\ref{the:OrchardTrivAut}). An advantage of orchard networks is that they can contain an unbounded number of reticulations, whereas tree-child networks can have at most $n-1$ reticulations, where $n$ is the number of leaves. Moreover, birth-hybridization priors that model hybridization as introgression (i.e., adding a horizontal edge from one existing lineage to another) naturally produce orchard networks.

The application of quotient Markov chains in Metropolis-Hastings in this paper all pertain to directed phylogenetic networks. However, the techniques can easily be applied to other classes of graphs or phylogenetic networks, such as undirected networks and semi-directed networks, which have attracted much attention lately. This in particular raises some questions semi-directed networks. For example, under which moves is the space of semi-directed networks connected and what is its diameter? This question has already been answered partly \cite{linz2023exploring}.

Another more novel question is on the complexity of \textsc{Graph Isomorphism} (and thus of \textsc{Count Automorphisms}) when restricted to semi-directed networks. We conjecture that this is GI-hard by a reduction from the problem on directed networks---which is GI-hard \cite{cardona2014comparison}---to the problem on semi-directed networks by replacing the root by a leaf and taking the underlying semi-directed network. Similarly, a simple reduction to \textsc{Graph Isomorphism} for directed networks by trying all root locations shows the problem is polynomial time solvable for binary semi-directed networks (using linear time orientation from \cite{huber2024orienting}). Can we attain a practical running time for all binary semi-directed networks, or is it necessary to consider subclasses of such networks?

Lastly, it would be interesting to see more connections to sampling methods for other types of graphs, such as graphs with a preset degree-sequence. For these kinds of spaces, the mixing times of certain Markov chains have been well studied (e.g., \cite{kleer2019nash}). Perhaps, improvements for results about mixing times can make use of the symmetric properties of the proposal Markov chain, too \cite[e.g.][]{doi:10.1137/070689413, doi:10.1080/15427951.2005.10129100}. Mixing times could give indications for good burn-in periods and expected time until convergence. Therefore, results about mixing times for sampling leaf-labelled networks could be very helpful for Bayesian methods.

\section{Code Availability}
The data of Figure~\ref{fig:undersampling_n2k2} was produced with PhyloX \cite{janssen2024phylox}, version v1.1.2. 
The source code of PhyloX can be found at \url{https://github.com/RemieJanssen/PhyloX} and in archived version at \url{https://doi.org/10.5281/zenodo.17228269}. 
The script for producing the data for the figure can be found at \url{https://github.com/RemieJanssen/test-phylogenetic-MCMC-sampling/releases/tag/v1.0.0}. 

The use of $\mu$-vectors to speed up counting automorphisms and checking isomorphism has been implemented in PhyloX since version v1.1.3. An archived version can be found at \url{https://doi.org/10.5281/zenodo.20447205}. The running time tests for these functions can be found at \url{https://github.com/RemieJanssen/test-phylox-isomorphisms/releases/tag/v1.0.0}.

% Acknowledgements
\section*{Acknowledgements}
Mark Jones was partially funded by the Dutch Research Council (NWO) grant OCENW.KLEIN.125.

\bibliographystyle{abbrvurl}
\bibliography{bibliography}

\appendix
\section{Sampling Networks}\label{sec:Appendix_Sampling}
To create the data for Figure~\ref{fig:undersampling_n2k2}, we have generated networks using PhyloNet and PhyloX.
To count the number of networks in the experiments, we calculate several properties of the generated networks.
Together, these properties determine which network in $\N(2,2)$ the network is isomorphic to.
In this section, we define these properties and calculate them for each network in $\N(2,2)$.
Then we discuss our sampling experiments and their results, which are summarized in Table~\ref{tab:network_properties}.
The code and raw results for all experiments are available at \url{https://github.com/RemieJanssen/test-phylonet-sampling} version \href{https://github.com/RemieJanssen/test-phylogenetic-MCMC-sampling/releases/tag/v1.0.0}{v1.0.0}.
All experiments were run on Red Hat Enterprise Linux 8 nodes of a computing cluster, with a limit of 1 CPU and 12GB memory.

\subsection{Network Properties}
The \textbf{$B_2$ index} of a network $N$ is a measure of its balance defined as $-\sum_{l\in L} p_l \log_2(p_l)$, where $L$ is the set of leaves of $N$ and $p_l$ is the probability of reaching leaf $l$ in a random walk from the root. \cite{bienvenu2021revisiting} The balance of a network can be computed by a single root-to-leaves traversal. 

The \textbf{blob sizes} of a network are the sizes (number of nodes) in each of the non-trivial biconnected components of the network. Note that each network in $\N(2,2)$ has exactly one non-trivial biconnected component, so we can talk about the blob size of the network in this setting.

The next few properties have to do with the placement of leaves in relation to other structures in the network, such as triangles and reticulations. We say that a leaf is \textbf{under a reticulation} if its parent is a reticulation node; a leaf \textbf{has a sibling reticulation} if its parent is a tree node whose other child is a reticulation node; a leaf is \textbf{on the side of a triangle} if its parent is a tree node $p$ with other child $r$ and parent $g$ such that there is an edge from $g$ to $r$; and a leaf is \textbf{on the bottom of a triangle} if the parent of the leaf is a reticulation with parents $g_1$ and $g_2$ such that there is an edge from $g_1$ to $g_2$ or the other way around. A pair of leaves $(l,m)$ forms a \textbf{reticulated cherry} if there is an edge from the the parent $p_m$ of $m$ to the parent $p_l$ of $l$ and $p_l$ is a reticulation node.  

\subsection{PhyloNet}
We sampled networks with two leaves and at most two reticulations using PhyloNet MCMC\_GT on PhyloNet version 3.8.2 \cite{wen2016bayesian,wen2018inferring}. By setting uninformative data, we can sample the prior distribution of the MCMC\_GT method. For this experiment, we have set the data to two gene trees, which both are the unique tree on two leaves. The full nexus file containing the configuration of the sampling is the following.
\begin{verbatim}
#NEXUS
BEGIN TREES;
Tree gt1 = (A,B);
Tree gt2 = (A,B);
END;

BEGIN PHYLONET;
MCMC_GT (gt1-gt2) -mr 2 -cl 10000500 -bl 500 -sf 500 -pp 1000 -pl 8 -sd 4321;
END;
\end{verbatim}
Samples were counted as described above, and the results can be found in Table~\ref{tab:network_properties}. The counts do not add up to the total sample for the simulation of 20000 networks, as we left out the sampled networks with one or zero reticulations. We were not able to fully explain the sampling frequencies we found, as PhyloNet's prior is quite involved, and does take into account branch lengths in several ways. Hence --- even though the only network with non-trivial symmetry (network id $A$ in the table) was sampled less than the rest --- it is uninformative to have a discussion about correcting for symmetries in this case. 

\subsection{PhyloX}
To verify the methods from this paper, we have implemented uniform MCMC sampling of phylogenetic networks with and without correction for symmetry. The implementation is part of the Python package PhyloX \cite{janssen2024phylox}, version v1.1.2 and above, available in archived version at Zenodo at \url{https://doi.org/10.5281/zenodo.17228269}.

A total of 10000 samples were taken after every 100 move proposals. The starting network was the network with id $A$ in Table~\ref{tab:network_properties}, and the move type proposal probabilities were: $0.2$ for a Head move, and $0.8$ for a Tail move. The sampling was performed with and without correction for symmetries. Results are shown in Table~\ref{tab:network_properties}. Note that the network with id $A$ is sampled about half as often compared to the other networks when not correcting for symmetries.

\begin{longtable}
{p{0.1\linewidth}|l||p{0.07\linewidth}|p{0.05\linewidth}|p{0.05\linewidth}|p{0.05\linewidth}|p{0.05\linewidth}|p{0.05\linewidth}|p{0.05\linewidth}||p{0.055\linewidth}|p{0.055\linewidth}|p{0.055\linewidth}}
     \scriptsize network& \scriptsize id& \scriptsize balance& \scriptsize blob size& \scriptsize leaf under a reticulation& \scriptsize leaf with sibling reticulation& \scriptsize leaf on the side of a triangle& \scriptsize leaf on the bottom of a triangle& \scriptsize reticu\-lated cherry& \scriptsize PhyloX corrected& \scriptsize PhyloX uncorrected& \scriptsize Phylo\-Net MCMC\-\_GT\\ \hline\hline
     \resizebox{\linewidth}{!}{\input{images/networks2-2/A.tex}}    &$A$&   1& 5& A,B& -& -& -&-&575&\textbf{255}&281\\ \hline
     \resizebox{\linewidth}{!}{\input{images/networks2-2/B_A.tex}}    &$B^A$&   0.544& 5& A& B& B& -&-&550&615&650\\ \hline
     \resizebox{\linewidth}{!}{\input{images/networks2-2/B_B.tex}}    &$B^B$&   0.544& 5& B& A& A& -&-&564&585&675\\ \hline
     \resizebox{\linewidth}{!}{\input{images/networks2-2/C_A.tex}}    &$C^A$&   0.954& 5& A,B& -& -& A&-&549&585&1226\\ \hline
     \resizebox{\linewidth}{!}{\input{images/networks2-2/C_B.tex}}    &$C^B$&   0.954& 5& A,B& -& -& B&-&562&585&1183\\ \hline
     \resizebox{\linewidth}{!}{\input{images/networks2-2/D1_A.tex}}    &$D_1^A$&   1& 4& A& -& -& A&-&559&547&3128\\ \hline
     \resizebox{\linewidth}{!}{\input{images/networks2-2/D1_B.tex}}    &$D_1^B$&   1& 4& B& -& -& B&-&561&565&3212\\ \hline
     \resizebox{\linewidth}{!}{\input{images/networks2-2/D2_A.tex}}    &$D_2^A$&   0.811& 5& A& -& -& A&-&542&591&253\\ \hline
     \resizebox{\linewidth}{!}{\input{images/networks2-2/D2_B.tex}}    &$D_2^B$&   0.811& 5& B& -& -& B&-&531&563&277\\ \hline
     \resizebox{\linewidth}{!}{\input{images/networks2-2/D3_A.tex}}    &$D_3^A$&   0.811& 5& A& B& -& A&-&551&592&565\\ \hline
     \resizebox{\linewidth}{!}{\input{images/networks2-2/D3_B.tex}}    &$D_3^B$&   0.811& 5& B& A& -& B&-&541&569&536\\ \hline
     \resizebox{\linewidth}{!}{\input{images/networks2-2/D4_A.tex}}    &$D_4^A$&   0.544& 5& A& B& -& -&-&551&549&422\\ \hline
     \resizebox{\linewidth}{!}{\input{images/networks2-2/D4_B.tex}}    &$D_4^B$&   0.544& 5& B& A& -& -&-&569&530&397\\ \hline
     \resizebox{\linewidth}{!}{\input{images/networks2-2/D5_A.tex}}    &$D_5^A$&   0.544& 5& A& B& -& -&(A,B)&565&554&573\\ \hline
     \resizebox{\linewidth}{!}{\input{images/networks2-2/D5_B.tex}}    &$D_5^B$&   0.544& 5& B& A& -& -&(B,A)&558&582&594\\ \hline
     \resizebox{\linewidth}{!}{\input{images/networks2-2/D6_A.tex}}    &$D_6^A$&   0.954& 5& A& B& -& -&(A,B)&553&587&747\\ \hline
     \resizebox{\linewidth}{!}{\input{images/networks2-2/D6_B.tex}}    &$D_6^B$&   0.954& 5& B& A& -& -&(B,A)&532&575&762\\ \hline
     \resizebox{\linewidth}{!}{\input{images/networks2-2/D7.tex}}    &$D_7^A$&   1& 4& -& -& -& -&-&591&567&2889\\ \hline
    \caption{Properties of all networks in $\N(2,2)$ to distinguish them, and sample counts for PhyloNet and PhyloX. Note in particular the PhyloX uncorrected sample count for network $A$, which is the only network in $\N(2,2)$ with non-trivial symmetries.}
    \label{tab:network_properties}
\end{longtable}

\section{Running Time Improvements Using \texorpdfstring{$\mu$}{mu}-Vectors}\label{sec:Appendix_Running_Time_Mu}

We have implemented the algorithms for counting automorphisms in Section~\ref{sec:counting_automorphisms} into PhyloX version v1.1.3 \cite{janssen2024phylox}. 
To test the efficiency of this code, we have compared running times of \texttt{phylox.isomorphism.is\_isomorphic} and \texttt{phylox.isomorphism.count\_automorphisms} when using $\mu$-vectors or not, 
which is controlled with the \texttt{use\_mu\_vector} parameter in both these functions.  

The code and raw results for all experiments are available at \url{https://github.com/RemieJanssen/test-phylox-isomorphisms} version \href{https://github.com/RemieJanssen/test-phylox-isomorphisms/releases/tag/v1.0.0}{v1.0.0}. All experiments were run on Red Hat Enterprise Linux 8 nodes of a computing cluster, with a limit of 1 CPU and 12GB memory.

\subsection{Isomorphism Checking}
For several combinations of numbers of leaves ($n$) and reticulations ($k$), we have generated $50$ networks at random using MCMC sampling. We did not correct for symmetries, because the likelihood of sampling a network with symmetries is quite low anyway, and the time to generate the networks would explode. The number of move proposals between each sample was set to $50(n+k)$. And each proposal was either a tail move with probability $0.8$ or a head move with probability $0.2$.

For each pair of networks in such a set, we have run and timed \texttt{phylox.isomorphism.is\_iso\-morphic} with and without the use of $\mu$-vectors. Inspection of the data revealed that the running times had a bimodal distribution. This was caused by a strong dependence on whether the two input networks were isomorphic or not. Hence, we have split the data into the categories `same network' with values `yes' or `no'. The running times are shown in Figure~\ref{fig:isom_results}.

\begin{figure}
    \centering
    \includegraphics[width=1.0\linewidth]{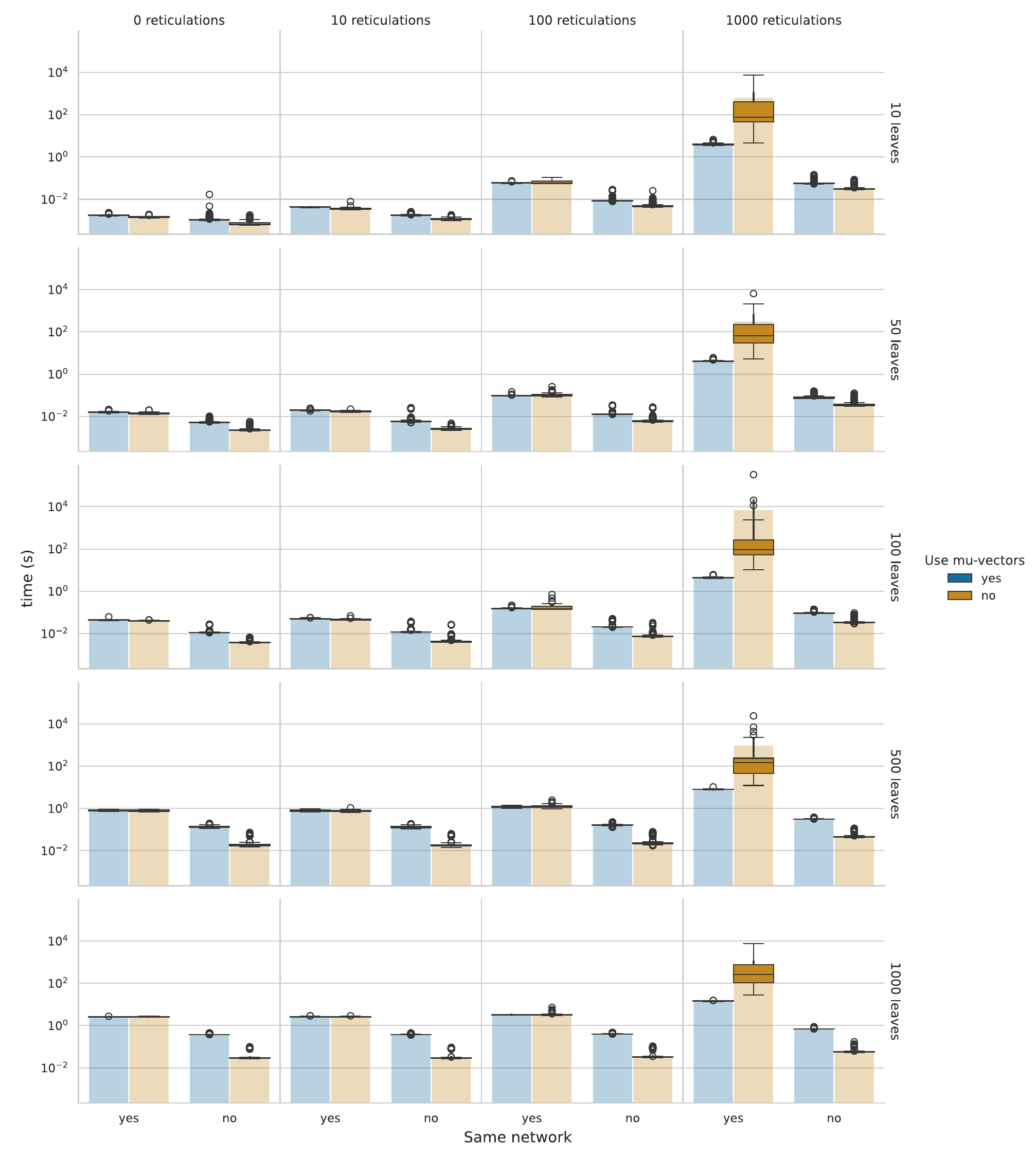}
    \caption{Running times of \texttt{phylox.isomorphism.is\_isomorphic} in seconds (log scale). Data are shown as box-plots with a bar chart showing the mean. Data are separated by $n$, $k$, ``same network'', and ``use $\mu$-vector''. The running times are faster for isomorphic networks using $\mu$-vectors, but for non-isomorphic networks when not using $\mu$-vectors. The running time for non-isomorphic networks is much faster than for isomorphic networks.}
    \label{fig:isom_results}
\end{figure}

In general, both methods perform much better on non-isomorphic networks than on isomorphic networks. For non-isomorphic networks, both methods perform well (within several seconds), although not-using $\mu$-vectors can be up to 10 times faster (Figure~\ref{fig:isom_speedup}, right).  

\begin{figure}
    \centering
    \includegraphics[width=1.0\linewidth]{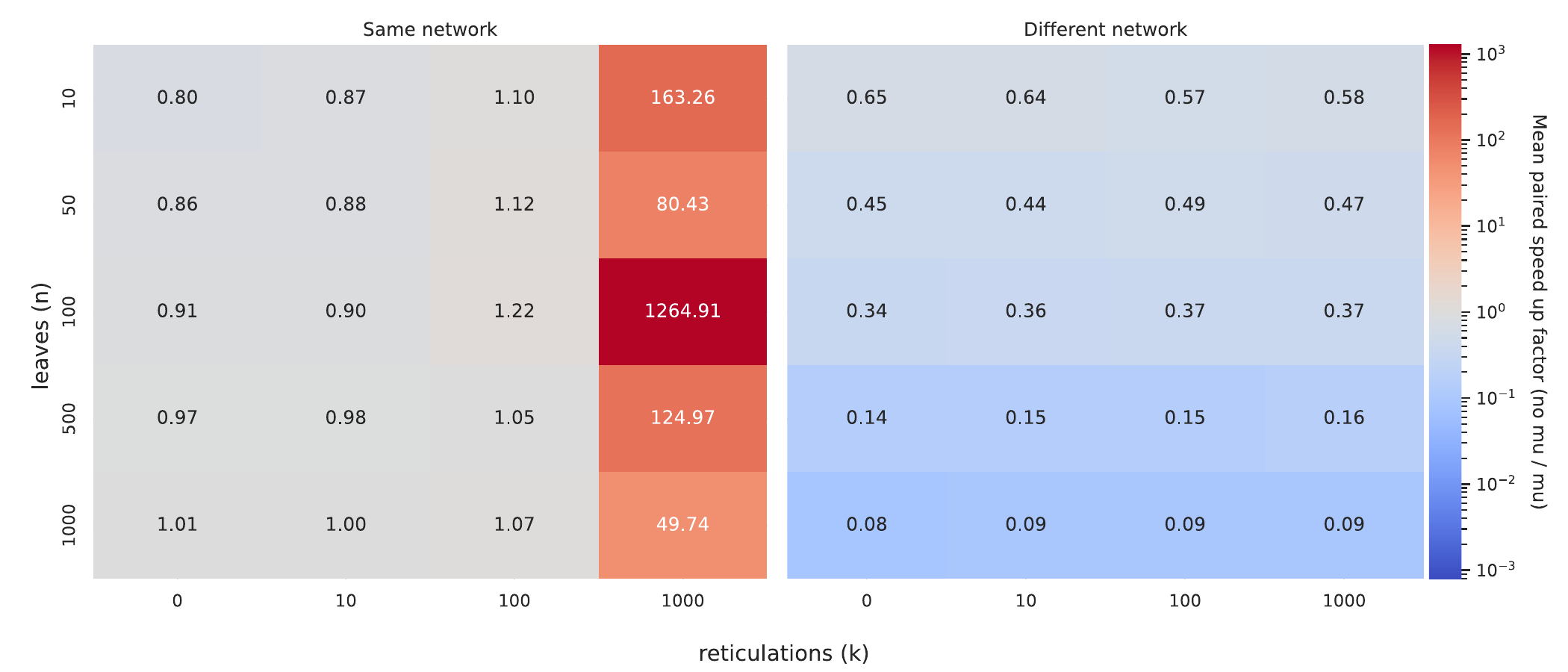}
    \caption{Mean running time improvement factor for \texttt{phylox.isomorphism.is\_isomorphic}. Pairwise running time factors (time without using $\mu$-vectors divided by time with using $\mu$-vectors) were computed and then averaged. A value of $1$ indicated the same running time when using or not using $\mu$-vectors. A value smaller than $1$ indicates a faster running time when not using $\mu$-vectors, and a value larger than $1$ indicates a faster running time when using $\mu$-vectors.}
    \label{fig:isom_speedup}
\end{figure}

Notably, when the networks are isomorphic and the number of reticulations is large, the running times for isomorphic networks differ strongly. In those cases, running times are much longer compared to non-isomorphic networks. Moreover, using $\mu$-vectors speeds up this computation by a large factor (Figure~\ref{fig:isom_speedup}, left). In one exceptional case ($n=100$, $k=1000$, network id 45), the running time using $\mu$-vectors was $6$ seconds, whereas the running time without using $\mu$-vectors was $327543$ seconds (i.e., about 91 hours).

In practice, this means there could be instances when not using $\mu$-vectors is faster. However, when handling larger networks that are expected to be isomorphic in some cases, it is likely that using $\mu$-vectors will produce results faster.

\subsection{Counting Automorphisms}
To test the effect of using $\mu$-vectors on the running time of \texttt{phylox.iso\-morphism.count\_auto\-morphisms}, we have generated a test set of networks similar to that described above. In this test set, the number of networks ($N$) for a given combination of $n$ and $k$ was not fixed to keep the total running time bounded. 

For each network in this test set, we have run and timed the execution of \texttt{phylox.iso\-morphism.count\_automorphisms} with and without the use of $\mu$-vectors. 
The running times increase with both $n$ and $k$, but seemingly faster with $k$. Running times are consistently faster when using $\mu$-vectors (Figure~\ref{fig:autom_results}). Using mu-vectors is on average about 25 times as fast for small networks ($n=10$), but is consistently over 100 times as fast for larger networks.

\begin{figure}
    \centering
    \includegraphics[width=1.0\linewidth]{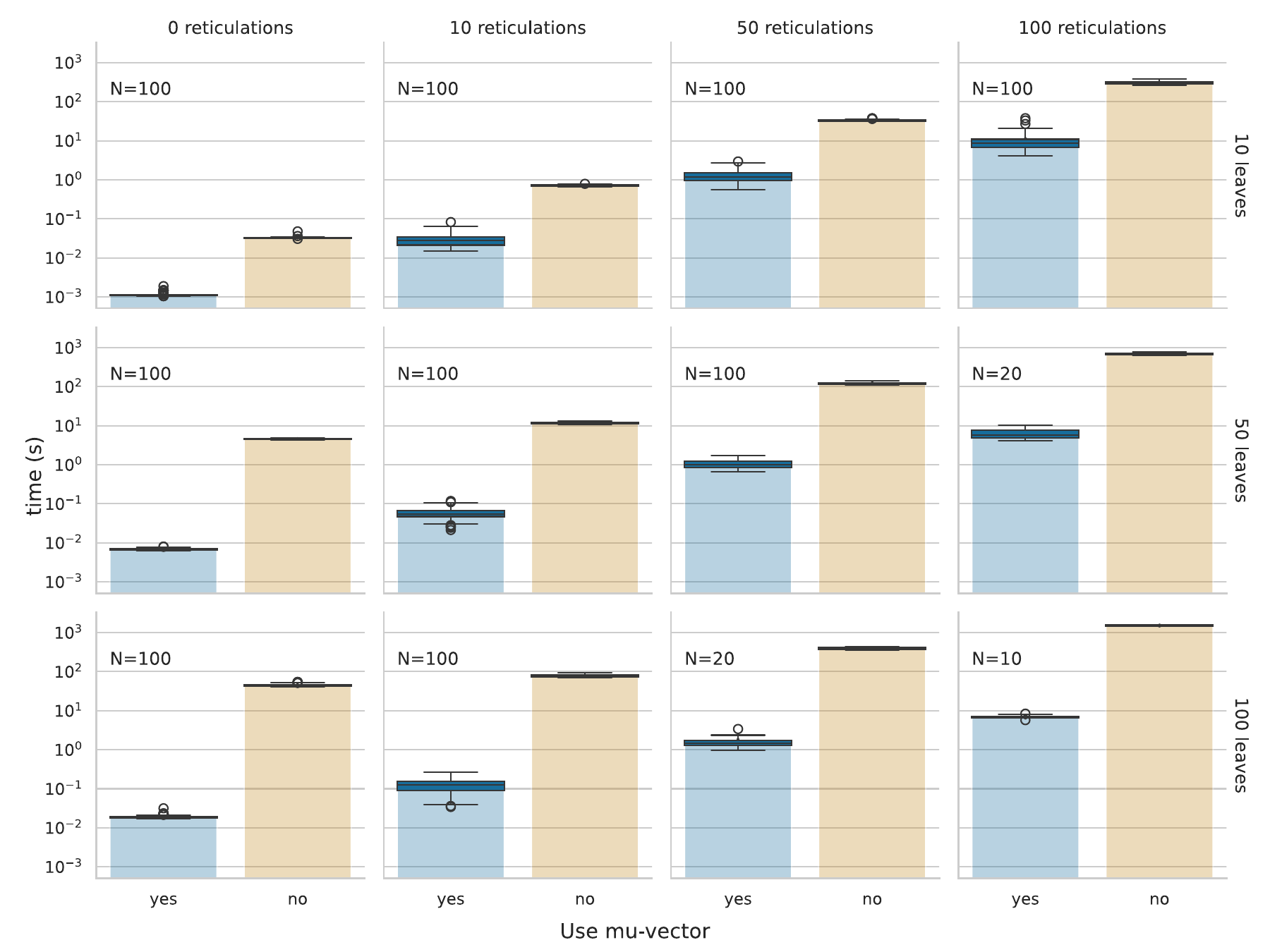}
    \caption{Running times of \texttt{phylox.isomorphism.count\_automorphisms} in seconds (log scale). Data are shown as box-plots with a bar chart showing the mean. The number of networks used for each plot ($N$) is shown in the plot. Data are separated by $n$, $k$, and ``use $\mu$-vector''. The running times increase with both $n$ and $k$, but seemingly faster with $k$. Running times are consistently faster when using $\mu$-vectors.}
    \label{fig:autom_results}
\end{figure}

\end{document}

%% file: images/networks2-2/A.tex
\begin{tikzpicture}[
    node distance=0.3cm and 1cm,
    every node/.style={circle,draw,minimum size=.5mm},
    leaf/.style={rectangle,draw=none, minimum size=6mm},
    >=Stealth
]

% Nodes
\node (pr) at (0,0) {}; % root
\node (r) [below=of pr] {};
\node (a) [below left=of r] {}; 
\node (b) [below right=of r] {};

% Reticulations
\node (r1) [below=of a] {}; 
\node (r2) [below=of b] {};

% Leaves
\node[leaf] (L1) [below=of r1] {A};
\node[leaf] (L2) [below=of r2] {B};

% Edges from root
\draw[->] (pr) -- (r);
\draw[->] (r) -- (a);
\draw[->] (r) -- (b);

% Edges to reticulations
\draw[->] (a) -- (r1);
\draw[->] (b) -- (r1); % reticulation 1
\draw[->] (b) -- (r2);
\draw[->] (a) -- (r2); % reticulation 2

% Edges to leaves
\draw[->] (r1) -- (L1);
\draw[->] (r2) -- (L2);

\end{tikzpicture}

%% file: images/networks2-2/B_A.tex
\begin{tikzpicture}[
    node distance=0.3cm and 1cm,
    every node/.style={circle,draw,minimum size=.5mm},
    leaf/.style={rectangle,draw=none, minimum size=6mm},
    >=Stealth
]

% Nodes
\node (pr) at (0,0) {}; % root
\node (r) [below=of pr] {};
\node (a) [below left=of r] {}; 

\node (b) [below left=of a] {}; 
\node (r2) [below right=of b] {};
\node (r1) [below right=of r2] {};

% Leaves
\node[leaf] (L2) [below=of r1] {B};
\node[leaf] (L1) [below=of b, left=of L2] {A};

% Edges from root
\draw[->] (pr) -- (r);
\draw[->] (r) -- (a);
\draw[->] (r) -- (r1);

\draw[->] (a) -- (b);
\draw[->] (a) -- (r2); % reticulation 1

\draw[->] (b) -- (r2);
\draw[->] (r2) -- (r1); % reticulation 2

% Edges to leaves
\draw[->] (b) -- (L1);
\draw[->] (r1) -- (L2);

\end{tikzpicture}

%% file: images/networks2-2/B_B.tex
\begin{tikzpicture}[
    node distance=0.3cm and 1cm,
    every node/.style={circle,draw,minimum size=.5mm},
    leaf/.style={rectangle,draw=none, minimum size=6mm},
    >=Stealth
]

% Nodes
\node (pr) at (0,0) {}; % root
\node (r) [below=of pr] {};
\node (a) [below left=of r] {}; 

\node (b) [below left=of a] {}; 
\node (r2) [below right=of b] {};
\node (r1) [below right=of r2] {};

% Leaves
\node[leaf] (L2) [below=of r1] {A};
\node[leaf] (L1) [below=of b, left=of L2] {B};

% Edges from root
\draw[->] (pr) -- (r);
\draw[->] (r) -- (a);
\draw[->] (r) -- (r1);

\draw[->] (a) -- (b);
\draw[->] (a) -- (r2); % reticulation 1

\draw[->] (b) -- (r2);
\draw[->] (r2) -- (r1); % reticulation 2

% Edges to leaves
\draw[->] (b) -- (L1);
\draw[->] (r1) -- (L2);

\end{tikzpicture}

%% file: images/networks2-2/C_A.tex
\begin{tikzpicture}[
    node distance=0.3cm and 1cm,
    every node/.style={circle,draw,minimum size=.5mm},
    leaf/.style={rectangle,draw=none, minimum size=6mm},
    >=Stealth
]

% Nodes
\node (pr) at (0,0) {}; % root
\node (r) [below=of pr] {};
\node (a) [below left=of r] {}; 

\node (b) [below right=of a] {}; 
\node (r1) [below left=of b] {};
\node (r2) [below right=of b] {};

% Leaves
\node[leaf] (L2) [below=of r2] {B};
\node[leaf] (L1) [below=of r1] {A};

% Edges from root
\draw[->] (pr) -- (r);
\draw[->] (r) -- (a);
\draw[->] (r) -- (r2);

\draw[->] (a) -- (b);
\draw[->] (a) -- (r1); % reticulation 1

\draw[->] (b) -- (r1);
\draw[->] (b) -- (r2); % reticulation 2

% Edges to leaves
\draw[->] (r1) -- (L1);
\draw[->] (r2) -- (L2);

\end{tikzpicture}

%% file: images/networks2-2/C_B.tex
\begin{tikzpicture}[
    node distance=0.3cm and 1cm,
    every node/.style={circle,draw,minimum size=.5mm},
    leaf/.style={rectangle,draw=none, minimum size=6mm},
    >=Stealth
]

% Nodes
\node (pr) at (0,0) {}; % root
\node (r) [below=of pr] {};
\node (a) [below left=of r] {}; 

\node (b) [below right=of a] {}; 
\node (r1) [below left=of b] {};
\node (r2) [below right=of b] {};

% Leaves
\node[leaf] (L2) [below=of r2] {A};
\node[leaf] (L1) [below=of r1] {B};

% Edges from root
\draw[->] (pr) -- (r);
\draw[->] (r) -- (a);
\draw[->] (r) -- (r2);

\draw[->] (a) -- (b);
\draw[->] (a) -- (r1); % reticulation 1

\draw[->] (b) -- (r1);
\draw[->] (b) -- (r2); % reticulation 2

% Edges to leaves
\draw[->] (r1) -- (L1);
\draw[->] (r2) -- (L2);

\end{tikzpicture}

%% file: images/networks2-2/D1_A.tex
\begin{tikzpicture}[
    node distance=0.3cm and 1cm,
    every node/.style={circle,draw,minimum size=.5mm},
    leaf/.style={rectangle,draw=none, minimum size=6mm},
    >=Stealth
]

% Nodes
\node (pr) at (0,0) {}; % root
\node (b) [below=of pr] {}; 
\node (r) [below=of b] {};
\node (a) [below left=of r] {}; 
\node (r1) [below right=of a] {};
\node (r2) [below left=of r1] {};

% Leaves
\node[leaf] (L1) [below left=of r2] {A};
\node[leaf] (L2) [below right=of b] {B};

% Edges from root
\draw[->] (b) -- (r);
\draw[->] (r) -- (a);
\draw[->] (r) -- (r1);
\draw[->] (a) -- (r1);
\draw[->] (a) -- (r2); % reticulation 1
\draw[->] (r1) -- (r2);
\draw[->] (pr) -- (b); % reticulation 2

% Edges to leaves
\draw[->] (r2) -- (L1);
\draw[->] (b) -- (L2);

\end{tikzpicture}

%% file: images/networks2-2/D1_B.tex
\begin{tikzpicture}[
    node distance=0.3cm and 1cm,
    every node/.style={circle,draw,minimum size=.5mm},
    leaf/.style={rectangle,draw=none, minimum size=6mm},
    >=Stealth
]

% Nodes
\node (pr) at (0,0) {}; % root
\node (b) [below=of pr] {}; 
\node (r) [below=of b] {};
\node (a) [below left=of r] {}; 
\node (r1) [below right=of a] {};
\node (r2) [below left=of r1] {};

% Leaves
\node[leaf] (L1) [below left=of r2] {B};
\node[leaf] (L2) [below right=of b] {A};

% Edges from root
\draw[->] (b) -- (r);
\draw[->] (r) -- (a);
\draw[->] (r) -- (r1);
\draw[->] (a) -- (r1);
\draw[->] (a) -- (r2); % reticulation 1
\draw[->] (r1) -- (r2);
\draw[->] (pr) -- (b); % reticulation 2

% Edges to leaves
\draw[->] (r2) -- (L1);
\draw[->] (b) -- (L2);

\end{tikzpicture}

%% file: images/networks2-2/D2_A.tex
\begin{tikzpicture}[
    node distance=0.3cm and 1cm,
    every node/.style={circle,draw,minimum size=.5mm},
    leaf/.style={rectangle,draw=none, minimum size=6mm},
    >=Stealth
]

% Nodes
\node (pr) at (0,0) {}; % root
\node (r) [below=of pr] {};
\node (b) [below left =of r] {}; 
\node (a) [below=of b] {}; 
\node (r1) [below right=of a] {};
\node (r2) [below left=of r1] {};

% Leaves
\node[leaf] (L1) [below=of r2] {A};
\node[leaf] (L2) [below left=of b] {B};

% Edges from root
\draw[->] (pr) -- (r);
\draw[->] (b) -- (a);
\draw[->] (r) -- (r1);
\draw[->] (a) -- (r1);
\draw[->] (a) -- (r2); % reticulation 1
\draw[->] (r1) -- (r2);
\draw[->] (r) -- (b); % reticulation 2

% Edges to leaves
\draw[->] (r2) -- (L1);
\draw[->] (b) -- (L2);

\end{tikzpicture}

%% file: images/networks2-2/D2_B.tex
\begin{tikzpicture}[
    node distance=0.3cm and 1cm,
    every node/.style={circle,draw,minimum size=.5mm},
    leaf/.style={rectangle,draw=none, minimum size=6mm},
    >=Stealth
]

% Nodes
\node (pr) at (0,0) {}; % root
\node (r) [below=of pr] {};
\node (b) [below left =of r] {}; 
\node (a) [below=of b] {}; 
\node (r1) [below right=of a] {};
\node (r2) [below left=of r1] {};

% Leaves
\node[leaf] (L1) [below=of r2] {B};
\node[leaf] (L2) [below left=of b] {A};

% Edges from root
\draw[->] (pr) -- (r);
\draw[->] (b) -- (a);
\draw[->] (r) -- (r1);
\draw[->] (a) -- (r1);
\draw[->] (a) -- (r2); % reticulation 1
\draw[->] (r1) -- (r2);
\draw[->] (r) -- (b); % reticulation 2

% Edges to leaves
\draw[->] (r2) -- (L1);
\draw[->] (b) -- (L2);

\end{tikzpicture}

%% file: images/networks2-2/D3_A.tex
\begin{tikzpicture}[
    node distance=0.3cm and 1cm,
    every node/.style={circle,draw,minimum size=.5mm},
    leaf/.style={rectangle,draw=none, minimum size=6mm},
    >=Stealth
]

% Nodes
\node (pr) at (0,0) {}; % root
\node (r) [below=of pr] {};
\node (a) [below left=of r] {}; 
\node (r1) [below right=of a] {};
\node (r2) [below left=of r1] {}; 

\node (b) [below right=of r] {}; 

% Leaves
\node[leaf] (L1) [below=of r2] {A};
\node[leaf] (L2) [below=of b] {B};

% Edges from root
\draw[->] (pr) -- (r);
\draw[->] (r) -- (a);
\draw[->] (r) -- (b);
\draw[->] (b) -- (r1);
\draw[->] (a) -- (r2); % reticulation 1
\draw[->] (r1) -- (r2);
\draw[->] (a) -- (r1); % reticulation 2

% Edges to leaves
\draw[->] (r2) -- (L1);
\draw[->] (b) -- (L2);

\end{tikzpicture}

%% file: images/networks2-2/D3_B.tex
\begin{tikzpicture}[
    node distance=0.3cm and 1cm,
    every node/.style={circle,draw,minimum size=.5mm},
    leaf/.style={rectangle,draw=none, minimum size=6mm},
    >=Stealth
]

% Nodes
\node (pr) at (0,0) {}; % root
\node (r) [below=of pr] {};
\node (a) [below left=of r] {}; 
\node (r1) [below right=of a] {};
\node (r2) [below left=of r1] {}; 

\node (b) [below right=of r] {}; 

% Leaves
\node[leaf] (L1) [below=of r2] {B};
\node[leaf] (L2) [below=of b] {A};

% Edges from root
\draw[->] (pr) -- (r);
\draw[->] (r) -- (a);
\draw[->] (r) -- (b);
\draw[->] (b) -- (r1);
\draw[->] (a) -- (r2); % reticulation 1
\draw[->] (r1) -- (r2);
\draw[->] (a) -- (r1); % reticulation 2

% Edges to leaves
\draw[->] (r2) -- (L1);
\draw[->] (b) -- (L2);

\end{tikzpicture}

%% file: images/networks2-2/D4_A.tex
\begin{tikzpicture}[
    node distance=0.3cm and 1cm,
    every node/.style={circle,draw,minimum size=.5mm},
    leaf/.style={rectangle,draw=none, minimum size=6mm},
    >=Stealth
]

% Nodes
\node (pr) at (0,0) {}; % root
\node (r) [below=of pr] {};
\node (a) [below left=of r] {}; 
\node (r1) [below right=of b] {};
\node (b) [below right=of a] {}; 
\node[leaf] (L2) [below=of b] {B};
\node (r2) [below=of L2] {}; 
\node[leaf] (L1) [below=of r2] {A};

% Edges from root
\draw[->] (pr) -- (r);
\draw[->] (r) -- (a);
\draw[->] (r) -- (r1);
\draw[->] (a) -- (r2);
\draw[->] (a) -- (b); % reticulation 1
\draw[->] (r1) -- (r2);
\draw[->] (b) -- (r1); % reticulation 2

% Edges to leaves
\draw[->] (r2) -- (L1);
\draw[->] (b) -- (L2);

\end{tikzpicture}

%% file: images/networks2-2/D4_B.tex
\begin{tikzpicture}[
    node distance=0.3cm and 1cm,
    every node/.style={circle,draw,minimum size=.5mm},
    leaf/.style={rectangle,draw=none, minimum size=6mm},
    >=Stealth
]

% Nodes
\node (pr) at (0,0) {}; % root
\node (r) [below=of pr] {};
\node (a) [below left=of r] {}; 
\node (r1) [below right=of b] {};
\node (b) [below right=of a] {}; 
\node[leaf] (L2) [below=of b] {A};
\node (r2) [below=of L2] {}; 
\node[leaf] (L1) [below=of r2] {B};

% Edges from root
\draw[->] (pr) -- (r);
\draw[->] (r) -- (a);
\draw[->] (r) -- (r1);
\draw[->] (a) -- (r2);
\draw[->] (a) -- (b); % reticulation 1
\draw[->] (r1) -- (r2);
\draw[->] (b) -- (r1); % reticulation 2

% Edges to leaves
\draw[->] (r2) -- (L1);
\draw[->] (b) -- (L2);

\end{tikzpicture}

%% file: images/networks2-2/D5_A.tex
\begin{tikzpicture}[
    node distance=0.3cm and 1cm,
    every node/.style={circle,draw,minimum size=.5mm},
    leaf/.style={rectangle,draw=none, minimum size=6mm},
    >=Stealth
]

% Nodes
\node (pr) at (0,0) {}; % root
\node (r) [below=of pr] {};
\node (a) [below left=of r] {}; 
\node (r1) [below right=of a] {};
\node (b) [below=of a] {}; 
\node (r2) [below=of b] {}; 

% Leaves
\node[leaf] (L1) [below=of r2] {A};
\node[leaf] (L2) [below left=of b] {B};

% Edges from root
\draw[->] (pr) -- (r);
\draw[->] (r) -- (a);
\draw[->] (r) -- (r1);
\draw[->] (a) -- (r1);
\draw[->] (a) -- (b); % reticulation 1
\draw[->] (r1) -- (r2);
\draw[->] (b) -- (r2); % reticulation 2

% Edges to leaves
\draw[->] (r2) -- (L1);
\draw[->] (b) -- (L2);

\end{tikzpicture}

%% file: images/networks2-2/D5_B.tex
\begin{tikzpicture}[
    node distance=0.3cm and 1cm,
    every node/.style={circle,draw,minimum size=.5mm},
    leaf/.style={rectangle,draw=none, minimum size=6mm},
    >=Stealth
]

% Nodes
\node (pr) at (0,0) {}; % root
\node (r) [below=of pr] {};
\node (a) [below left=of r] {}; 
\node (r1) [below right=of a] {};
\node (b) [below=of a] {}; 
\node (r2) [below=of b] {}; 

% Leaves
\node[leaf] (L1) [below=of r2] {B};
\node[leaf] (L2) [below left=of b] {A};

% Edges from root
\draw[->] (pr) -- (r);
\draw[->] (r) -- (a);
\draw[->] (r) -- (r1);
\draw[->] (a) -- (r1);
\draw[->] (a) -- (b); % reticulation 1
\draw[->] (r1) -- (r2);
\draw[->] (b) -- (r2); % reticulation 2

% Edges to leaves
\draw[->] (r2) -- (L1);
\draw[->] (b) -- (L2);

\end{tikzpicture}

%% file: images/networks2-2/D6_A.tex
\begin{tikzpicture}[
    node distance=0.3cm and 1cm,
    every node/.style={circle,draw,minimum size=.5mm},
    leaf/.style={rectangle,draw=none, minimum size=6mm},
    >=Stealth
]

% Nodes
\node (pr) at (0,0) {}; % root
\node (r) [below=of pr] {};
\node (a) [below left=of r] {}; 
\node (r1) [below right=of a] {};
\node (b) [below=of r1] {}; 
\node (r2) [below left=of b] {}; 

% Leaves
\node[leaf] (L1) [below=of r2] {A};
\node[leaf] (L2) [below right=of b, right of=L1] {B};

% Edges from root
\draw[->] (pr) -- (r);
\draw[->] (r) -- (a);
\draw[->] (r) -- (r1);
\draw[->] (a) -- (r1);
\draw[->] (a) -- (r2); % reticulation 1
\draw[->] (r1) -- (b);
\draw[->] (b) -- (r2); % reticulation 2

% Edges to leaves
\draw[->] (r2) -- (L1);
\draw[->] (b) -- (L2);

\end{tikzpicture}

%% file: images/networks2-2/D6_B.tex
\begin{tikzpicture}[
    node distance=0.3cm and 1cm,
    every node/.style={circle,draw,minimum size=.5mm},
    leaf/.style={rectangle,draw=none, minimum size=6mm},
    >=Stealth
]

% Nodes
\node (pr) at (0,0) {}; % root
\node (r) [below=of pr] {};
\node (a) [below left=of r] {}; 
\node (r1) [below right=of a] {};
\node (b) [below=of r1] {}; 
\node (r2) [below left=of b] {}; 

% Leaves
\node[leaf] (L1) [below=of r2] {B};
\node[leaf] (L2) [below right=of b, right of=L1] {A};

% Edges from root
\draw[->] (pr) -- (r);
\draw[->] (r) -- (a);
\draw[->] (r) -- (r1);
\draw[->] (a) -- (r1);
\draw[->] (a) -- (r2); % reticulation 1
\draw[->] (r1) -- (b);
\draw[->] (b) -- (r2); % reticulation 2

% Edges to leaves
\draw[->] (r2) -- (L1);
\draw[->] (b) -- (L2);

\end{tikzpicture}

%% file: images/networks2-2/D7.tex
\begin{tikzpicture}[
    node distance=0.3cm and 1cm,
    every node/.style={circle,draw,minimum size=.5mm},
    leaf/.style={rectangle,draw=none, minimum size=6mm},
    >=Stealth
]

% Nodes
\node (pr) at (0,0) {}; % root
\node (r) [below=of pr] {};
\node (a) [below left=of r] {}; 
\node (r1) [below right=of a] {};
\node (r2) [below left=of r1] {}; 

\node (b) [below=of r2] {}; 

% Leaves
\node[leaf] (L1) [below left=of b] {A};
\node[leaf] (L2) [below right=of b] {B};

% Edges from root
\draw[->] (pr) -- (r);
\draw[->] (r) -- (a);
\draw[->] (r) -- (r1);
\draw[->] (a) -- (r1);
\draw[->] (a) -- (r2); % reticulation 1
\draw[->] (r1) -- (r2);
\draw[->] (r2) -- (b); % reticulation 2

% Edges to leaves
\draw[->] (b) -- (L1);
\draw[->] (b) -- (L2);

\end{tikzpicture}

%% file: main.bbl
\begin{thebibliography}{10}

\bibitem{arnason2018whole}
{\'U}.~{\'A}rnason, F.~Lammers, V.~Kumar, M.~A. Nilsson, and A.~Janke.
\newblock Whole-genome sequencing of the blue whale and other rorquals finds signatures for introgressive gene flow.
\newblock {\em Science advances}, 4(4):eaap9873, 2018.
\newblock \href {https://doi.org/10.1126/sciadv.aap9873} {\path{doi:10.1126/sciadv.aap9873}}.

\bibitem{bai2021defining}
A.~Bai, P.~L. Erd{\H{o}}s, C.~Semple, and M.~Steel.
\newblock Defining phylogenetic networks using ancestral profiles.
\newblock {\em Mathematical Biosciences}, 332:108537, 2021.
\newblock \href {https://doi.org/10.1016/j.mbs.2021.108537} {\path{doi:10.1016/j.mbs.2021.108537}}.

\bibitem{bapteste2013networks}
E.~Bapteste, L.~van Iersel, A.~Janke, S.~Kelchner, S.~Kelk, J.~O. McInerney, D.~A. Morrison, L.~Nakhleh, M.~Steel, L.~Stougie, et~al.
\newblock Networks: expanding evolutionary thinking.
\newblock {\em Trends in Genetics}, 29(8):439--441, 2013.
\newblock \href {https://doi.org/10.1016/j.tig.2013.05.007} {\path{doi:10.1016/j.tig.2013.05.007}}.

\bibitem{beals1999finding}
R.~Beals, R.~Chang, W.~Gasarch, and J.~Tor{\'a}n.
\newblock On finding the number of graph automorphisms.
\newblock {\em Chicago J. Theor. Comput. Sci}, 1999.
\newblock \href {https://doi.org/10.1109/SCT.1995.514867} {\path{doi:10.1109/SCT.1995.514867}}.

\bibitem{bienvenu2021revisiting}
F.~Bienvenu, G.~Cardona, and C.~Scornavacca.
\newblock Revisiting shao and sokal’s b2 index of phylogenetic balance.
\newblock {\em Journal of Mathematical Biology}, 83(5):52, 2021.
\newblock \href {https://doi.org/10.1007/s00285-021-01662-7} {\path{doi:10.1007/s00285-021-01662-7}}.

\bibitem{bordewich2017lost}
M.~Bordewich, S.~Linz, and C.~Semple.
\newblock Lost in space? {Generalising} subtree prune and regraft to spaces of phylogenetic networks.
\newblock {\em Journal of theoretical biology}, 423:1--12, 2017.
\newblock \href {https://doi.org/10.1016/j.jtbi.2017.03.032} {\path{doi:10.1016/j.jtbi.2017.03.032}}.

\bibitem{bordewich2007computing}
M.~Bordewich and C.~Semple.
\newblock Computing the minimum number of hybridization events for a consistent evolutionary history.
\newblock {\em Discrete Applied Mathematics}, 155(8):914--928, 2007.
\newblock \href {https://doi.org/10.1016/j.dam.2006.08.008} {\path{doi:10.1016/j.dam.2006.08.008}}.

\bibitem{bordewich2016determining}
M.~Bordewich and C.~Semple.
\newblock Determining phylogenetic networks from inter-taxa distances.
\newblock {\em Journal of mathematical biology}, 73(2):283--303, 2016.
\newblock \href {https://doi.org/10.1007/s00285-015-0950-8} {\path{doi:10.1007/s00285-015-0950-8}}.

\bibitem{doi:10.1137/070689413}
S.~Boyd, P.~Diaconis, P.~Parrilo, and L.~Xiao.
\newblock Fastest mixing markov chain on graphs with symmetries.
\newblock {\em SIAM Journal on Optimization}, 20(2):792--819, 2009.
\newblock \href {https://arxiv.org/abs/https://doi.org/10.1137/070689413} {\path{arXiv:https://doi.org/10.1137/070689413}}, \href {https://doi.org/10.1137/070689413} {\path{doi:10.1137/070689413}}.

\bibitem{bryant2017quirks}
C.~Bryant, M.~Fischer, S.~Linz, and C.~Semple.
\newblock On the quirks of maximum parsimony and likelihood on phylogenetic networks.
\newblock {\em Journal of theoretical biology}, 417:100--108, 2017.
\newblock \href {https://doi.org/10.1016/j.jtbi.2017.01.013} {\path{doi:10.1016/j.jtbi.2017.01.013}}.

\bibitem{cardona2008distance}
G.~Cardona, M.~Llabr{\'e}s, F.~Rossell{\'o}, and G.~Valiente.
\newblock A distance metric for a class of tree-sibling phylogenetic networks.
\newblock {\em Bioinformatics}, 24(13):1481--1488, 2008.
\newblock \href {https://doi.org/10.1093/bioinformatics/btn231} {\path{doi:10.1093/bioinformatics/btn231}}.

\bibitem{cardona2014comparison}
G.~Cardona, M.~Llabr{\'e}s, F.~Rossell{\'o}, and G.~Valiente.
\newblock The comparison of tree-sibling time consistent phylogenetic networks is graph isomorphism-complete.
\newblock {\em The Scientific World Journal}, 2014, 2014.
\newblock \href {https://doi.org/10.1155/2014/254279} {\path{doi:10.1155/2014/254279}}.

\bibitem{cardona2009comparison}
G.~Cardona, F.~Rossello, and G.~Valiente.
\newblock Comparison of tree-child phylogenetic networks.
\newblock {\em IEEE/ACM Transactions on Computational Biology and Bioinformatics}, 6(4):552--569, 2009.
\newblock \href {https://doi.org/10.1109/TCBB.2007.70270} {\path{doi:10.1109/TCBB.2007.70270}}.

\bibitem{doi:10.1080/00031305.1995.10476177}
S.~Chib and E.~Greenberg.
\newblock Understanding the metropolis-hastings algorithm.
\newblock {\em The American Statistician}, 49(4):327--335, 1995.
\newblock URL: \url{https://www.tandfonline.com/doi/abs/10.1080/00031305.1995.10476177}, \href {https://arxiv.org/abs/https://www.tandfonline.com/doi/pdf/10.1080/00031305.1995.10476177} {\path{arXiv:https://www.tandfonline.com/doi/pdf/10.1080/00031305.1995.10476177}}, \href {https://doi.org/10.1080/00031305.1995.10476177} {\path{doi:10.1080/00031305.1995.10476177}}.

\bibitem{chor2006finding}
B.~Chor and T.~Tuller.
\newblock Finding a maximum likelihood tree is hard.
\newblock {\em Journal of the ACM (JACM)}, 53(5):722--744, 2006.
\newblock \href {https://doi.org/10.1145/1183907.1183909} {\path{doi:10.1145/1183907.1183909}}.

\bibitem{cordella2001improved}
L.~P. Cordella, P.~Foggia, C.~Sansone, M.~Vento, et~al.
\newblock An improved algorithm for matching large graphs.
\newblock In {\em 3rd IAPR-TC15 workshop on graph-based representations in pattern recognition}, pages 149--159, 2001.
\newblock URL: \url{https://api.semanticscholar.org/CorpusID:15968654}.

\bibitem{erdHos2019class}
P.~L. Erd{\H{o}}s, C.~Semple, and M.~Steel.
\newblock A class of phylogenetic networks reconstructable from ancestral profiles.
\newblock {\em Mathematical biosciences}, 313:33--40, 2019.
\newblock \href {https://doi.org/10.1016/j.mbs.2019.04.009} {\path{doi:10.1016/j.mbs.2019.04.009}}.

\bibitem{ERDOS2021205}
P.~L. Erdős, A.~Francis, and T.~R. Mezei.
\newblock Rooted nni moves and distance-1 tail moves on tree-based phylogenetic networks.
\newblock {\em Discrete Applied Mathematics}, 294:205--213, 2021.
\newblock URL: \url{https://www.sciencedirect.com/science/article/pii/S0166218X21000597}, \href {https://doi.org/10.1016/j.dam.2021.02.016} {\path{doi:10.1016/j.dam.2021.02.016}}.

\bibitem{forster2020phylogenetic}
P.~Forster, L.~Forster, C.~Renfrew, and M.~Forster.
\newblock Phylogenetic network analysis of sars-cov-2 genomes.
\newblock {\em Proceedings of the National Academy of Sciences}, 117(17):9241--9243, 2020.
\newblock \href {https://doi.org/10.1073/pnas.2004999117} {\path{doi:10.1073/pnas.2004999117}}.

\bibitem{foulds1982steiner}
L.~R. Foulds and R.~L. Graham.
\newblock The steiner problem in phylogeny is np-complete.
\newblock {\em Advances in Applied mathematics}, 3(1):43--49, 1982.
\newblock \href {https://doi.org/10.1016/S0196-8858(82)80004-3} {\path{doi:10.1016/S0196-8858(82)80004-3}}.

\bibitem{francis2017bounds}
A.~Francis, K.~T. Huber, V.~Moulton, and T.~Wu.
\newblock Bounds for phylogenetic network space metrics.
\newblock {\em Journal of Mathematical Biology}, Aug 2017.
\newblock \href {https://doi.org/10.1007/s00285-017-1171-0} {\path{doi:10.1007/s00285-017-1171-0}}.

\bibitem{fuchs2019counting}
M.~Fuchs, B.~Gittenberger, and M.~Mansouri.
\newblock Counting phylogenetic networks with few reticulation vertices: tree-child and normal networks.
\newblock {\em Australasian Journal of Combinatorics}, 73(2):385--423, 2019.
\newblock \href {https://doi.org/10.1016/j.dam.2022.03.026} {\path{doi:10.1016/j.dam.2022.03.026}}.

\bibitem{gambette2017rearrangement}
P.~Gambette, L.~van Iersel, M.~Jones, M.~Lafond, F.~Pardi, and C.~Scornavacca.
\newblock Rearrangement moves on rooted phylogenetic networks.
\newblock {\em PLoS computational biology}, 13(8):e1005611, 2017.
\newblock \href {https://doi.org/10.1371/journal.pcbi.1005611} {\path{doi:10.1371/journal.pcbi.1005611}}.

\bibitem{gavryushkin2016space}
A.~Gavryushkin and A.~J. Drummond.
\newblock The space of ultrametric phylogenetic trees.
\newblock {\em Journal of theoretical biology}, 403:197--208, 2016.
\newblock \href {https://doi.org/10.1016/j.jtbi.2016.05.001} {\path{doi:10.1016/j.jtbi.2016.05.001}}.

\bibitem{GREENHILL20181}
C.~Greenhill and M.~Sfragara.
\newblock The switch markov chain for sampling irregular graphs and digraphs.
\newblock {\em Theoretical Computer Science}, 719:1--20, 2018.
\newblock URL: \url{https://www.sciencedirect.com/science/article/pii/S030439751730840X}, \href {https://doi.org/10.1016/j.tcs.2017.11.010} {\path{doi:10.1016/j.tcs.2017.11.010}}.

\bibitem{gross2020distinguishing}
E.~Gross, L.~van Iersel, R.~Janssen, M.~Jones, C.~Long, and Y.~Murakami.
\newblock Distinguishing level-1 phylogenetic networks on the basis of data generated by markov processes.
\newblock {\em Journal of Mathematical Biology}, 83(3):32, 2021.
\newblock \href {https://doi.org/10.1007/s00285-021-01653-8} {\path{doi:10.1007/s00285-021-01653-8}}.

\bibitem{networkx}
A.~A. Hagberg, D.~A. Schult, and P.~J. Swart.
\newblock Exploring network structure, dynamics, and function using networkx.
\newblock In G.~Varoquaux, T.~Vaught, and J.~Millman, editors, {\em Proceedings of the 7th Python in Science Conference}, pages 11 -- 15, Pasadena, CA USA, 2008.
\newblock \href {https://doi.org/10.25080/TCWV9851} {\path{doi:10.25080/TCWV9851}}.

\bibitem{hayamizu2021structure}
M.~Hayamizu.
\newblock A structure theorem for rooted binary phylogenetic networks and its implications for tree-based networks.
\newblock {\em SIAM Journal on Discrete Mathematics}, 35(4):2490--2516, 2021.
\newblock \href {https://doi.org/10.1137/19M1297403} {\path{doi:10.1137/19M1297403}}.

\bibitem{huber2016spaces}
K.~T. Huber, S.~Linz, V.~Moulton, and T.~Wu.
\newblock Spaces of phylogenetic networks from generalized nearest-neighbor interchange operations.
\newblock {\em Journal of Mathematical Biology}, 72(3):699--725, 2016.
\newblock \href {https://doi.org/10.1007/s00285-015-0899-7} {\path{doi:10.1007/s00285-015-0899-7}}.

\bibitem{huber2018quarnet}
K.~T. Huber, V.~Moulton, C.~Semple, and T.~Wu.
\newblock Quarnet inference rules for level-1 networks.
\newblock {\em Bulletin of mathematical biology}, 80(8):2137--2153, 2018.
\newblock \href {https://doi.org/10.1007/s11538-018-0450-2} {\path{doi:10.1007/s11538-018-0450-2}}.

\bibitem{huber2016transforming}
K.~T. Huber, V.~Moulton, and T.~Wu.
\newblock Transforming phylogenetic networks: Moving beyond tree space.
\newblock {\em Journal of Theoretical Biology}, 404:30--39, 2016.
\newblock \href {https://doi.org/10.1016/j.jtbi.2016.05.030} {\path{doi:10.1016/j.jtbi.2016.05.030}}.

\bibitem{huber2021reconstructibility}
K.~T. Huber, L.~van Iersel, R.~Janssen, M.~Jones, V.~Moulton, and Y.~Murakami.
\newblock Level-2 networks from shortest and longest distances.
\newblock {\em Discrete Applied Mathematics}, 306:138--165, 2022.
\newblock \href {https://doi.org/10.1016/j.dam.2021.09.026} {\path{doi:10.1016/j.dam.2021.09.026}}.

\bibitem{huber2024orienting}
K.~T. Huber, L.~van Iersel, R.~Janssen, M.~Jones, V.~Moulton, Y.~Murakami, and C.~Semple.
\newblock Orienting undirected phylogenetic networks.
\newblock {\em Journal of Computer and System Sciences}, 140:103480, 2024.
\newblock \href {https://doi.org/10.1016/j.jcss.2023.103480} {\path{doi:10.1016/j.jcss.2023.103480}}.

\bibitem{huber2017reconstructing}
K.~T. Huber, L.~Van~Iersel, V.~Moulton, C.~Scornavacca, and T.~Wu.
\newblock Reconstructing phylogenetic level-1 networks from nondense binet and trinet sets.
\newblock {\em Algorithmica}, 77(1):173--200, 2017.
\newblock \href {https://doi.org/10.1007/s00453-015-0069-8} {\path{doi:10.1007/s00453-015-0069-8}}.

\bibitem{janssen2018heading}
R.~Janssen.
\newblock Heading in the right direction? using head moves to traverse phylogenetic network space.
\newblock {\em Journal of Graph Algorithms and Applications}, 25(1):263--310, 2021.
\newblock \href {https://doi.org/10.7155/jgaa.00559} {\path{doi:10.7155/jgaa.00559}}.

\bibitem{thesis_janssen}
R.~Janssen.
\newblock {\em Rearranging phylogenetic networks}.
\newblock PhD thesis, Delft University of Technology, 2021.
\newblock \href {https://doi.org/10.4233/uuid:1b713961-4e6d-4bb5-a7d0-37279084ee57} {\path{doi:10.4233/uuid:1b713961-4e6d-4bb5-a7d0-37279084ee57}}.

\bibitem{janssen2024phylox}
R.~Janssen.
\newblock Phylox: A python package for complete phylogenetic network workflows.
\newblock {\em Journal of Open Source Software}, 9(103):6427, 2024.
\newblock \href {https://doi.org/10.21105/joss.06427} {\path{doi:10.21105/joss.06427}}.

\bibitem{janssen2018exploring}
R.~Janssen, M.~Jones, P.~L. Erd{\H{o}}s, L.~Van~Iersel, and C.~Scornavacca.
\newblock Exploring the tiers of rooted phylogenetic network space using tail moves.
\newblock {\em Bulletin of mathematical biology}, 80(8):2177--2208, 2018.
\newblock \href {https://doi.org/10.1007/s11538-018-0452-0} {\path{doi:10.1007/s11538-018-0452-0}}.

\bibitem{janssen2019rearrangement}
R.~Janssen and J.~Klawitter.
\newblock Rearrangement operations on unrooted phylogenetic networks.
\newblock {\em Theory and Applications of Graphs}, 6(2), 2019.
\newblock \href {https://doi.org/10.20429/tag.2019.060206} {\path{doi:10.20429/tag.2019.060206}}.

\bibitem{janssen2021cherry}
R.~Janssen and Y.~Murakami.
\newblock On cherry-picking and network containment.
\newblock {\em Theoretical Computer Science}, 856:121--150, 2021.
\newblock \href {https://doi.org/On cherry-picking and network containment} {\path{doi:On cherry-picking and network containment}}.

\bibitem{jin2006maximum}
G.~Jin, L.~Nakhleh, S.~Snir, and T.~Tuller.
\newblock Maximum likelihood of phylogenetic networks.
\newblock {\em Bioinformatics}, 22(21):2604--2611, 2006.
\newblock \href {https://doi.org/10.1093/bioinformatics/btl452} {\path{doi:10.1093/bioinformatics/btl452}}.

\bibitem{Klawitter2018SNPRneigh}
J.~{Klawitter}.
\newblock The snpr neighbourhood of tree-child networks.
\newblock {\em Journal of Graph Algorithms and Applications}, 22(2):329--355, 2018.
\newblock \href {https://doi.org/10.7155/jgaa.00472} {\path{doi:10.7155/jgaa.00472}}.

\bibitem{klawitter2020spaces}
J.~Klawitter.
\newblock {\em Spaces of phylogenetic networks}.
\newblock PhD thesis, University of Auckland, 2020.

\bibitem{kleer2019nash}
P.~Kleer.
\newblock {\em When Nash met Markov: Novel results for pure Nash equilibria and the switch Markov chain}.
\newblock PhD thesis, Vrije Universiteit, Amsterdam, 2019.
\newblock URL: \url{https://hdl.handle.net/1871/56131}.

\bibitem{kobler2012graph}
J.~Kobler, U.~Sch{\"o}ning, and J.~Tor{\'a}n.
\newblock {\em The graph isomorphism problem: its structural complexity}.
\newblock Springer Science \& Business Media, 1993.
\newblock \href {https://doi.org/10.1007/978-1-4612-0333-9} {\path{doi:10.1007/978-1-4612-0333-9}}.

\bibitem{linz2023exploring}
S.~Linz and K.~Wicke.
\newblock Exploring spaces of semi-directed level-1 networks.
\newblock {\em Journal of Mathematical Biology}, 87(5):70, 2023.
\newblock \href {https://doi.org/10.1007/s00285-023-02004-5} {\path{doi:10.1007/s00285-023-02004-5}}.

\bibitem{luks1982isomorphism}
E.~M. Luks.
\newblock Isomorphism of graphs of bounded valence can be tested in polynomial time.
\newblock {\em Journal of computer and system sciences}, 25(1):42--65, 1982.
\newblock \href {https://doi.org/10.1016/0022-0000(82)90009-5} {\path{doi:10.1016/0022-0000(82)90009-5}}.

\bibitem{markin2019robinson}
A.~Markin, T.~K. Anderson, V.~S. K.~T. Vadali, and O.~Eulenstein.
\newblock Robinson-foulds reticulation networks.
\newblock In {\em Proceedings of the 10th ACM International Conference on Bioinformatics, Computational Biology and Health Informatics}, pages 77--86, 2019.
\newblock \href {https://doi.org/10.1145/3307339.3342151} {\path{doi:10.1145/3307339.3342151}}.

\bibitem{mathon1979note}
R.~Mathon.
\newblock A note on the graph isomorphism counting problem.
\newblock {\em Information Processing Letters}, 8(3):131--136, 1979.
\newblock \href {https://doi.org/10.1016/0020-0190(79)90004-8} {\path{doi:10.1016/0020-0190(79)90004-8}}.

\bibitem{mcdiarmid2015counting}
C.~McDiarmid, C.~Semple, and D.~Welsh.
\newblock Counting phylogenetic networks.
\newblock {\em Annals of Combinatorics}, 19(1):205--224, 2015.
\newblock \href {https://doi.org/10.1007/s00026-015-0260-2} {\path{doi:10.1007/s00026-015-0260-2}}.

\bibitem{mena2012ternary}
A.~A. Mena and F.~Rossell{\'o}.
\newblock Ternary graph isomorphism in polynomial time, after luks.
\newblock {\em arXiv preprint arXiv:1209.0871}, 2012.
\newblock \href {https://doi.org/10.48550/arXiv.1209.0871} {\path{doi:10.48550/arXiv.1209.0871}}.

\bibitem{mendes2025validate}
F.~K. Mendes, R.~Bouckaert, L.~M. Carvalho, and A.~J. Drummond.
\newblock How to validate a bayesian evolutionary model.
\newblock {\em Systematic Biology}, 74(1):158--175, 2025.
\newblock \href {https://doi.org/10.1093/sysbio/syae064} {\path{doi:10.1093/sysbio/syae064}}.

\bibitem{mitavskiy2008quotients}
B.~Mitavskiy, J.~E. Rowe, A.~Wright, and L.~M. Schmitt.
\newblock Quotients of markov chains and asymptotic properties of the stationary distribution of the markov chain associated to an evolutionary algorithm.
\newblock {\em Genetic Programming and Evolvable Machines}, 9(2):109--123, 2008.
\newblock \href {https://doi.org/10.1007/s10710-007-9038-6} {\path{doi:10.1007/s10710-007-9038-6}}.

\bibitem{murakami2021thesis}
Y.~Murakami.
\newblock {\em On Phylogenetic Encodings and Orchard Networks}.
\newblock Dissertation (tu delft), Delft University of Technology, 2021.
\newblock \href {https://doi.org/10.4233/uuid:049932ab-4124-4639-a7e3-146ac4fd805d} {\path{doi:10.4233/uuid:049932ab-4124-4639-a7e3-146ac4fd805d}}.

\bibitem{nascimento2017biologist}
F.~F. Nascimento, M.~d. Reis, and Z.~Yang.
\newblock A biologist’s guide to bayesian phylogenetic analysis.
\newblock {\em Nature ecology \& evolution}, 1(10):1446--1454, 2017.
\newblock \href {https://doi.org/10.1038/s41559-017-0280-x} {\path{doi:10.1038/s41559-017-0280-x}}.

\bibitem{oldman2016trilonet}
J.~Oldman, T.~Wu, L.~Van~Iersel, and V.~Moulton.
\newblock Trilonet: piecing together small networks to reconstruct reticulate evolutionary histories.
\newblock {\em Molecular biology and evolution}, 33(8):2151--2162, 2016.
\newblock \href {https://doi.org/10.1093/molbev/msw068} {\path{doi:10.1093/molbev/msw068}}.

\bibitem{roch2006short}
S.~Roch.
\newblock A short proof that phylogenetic tree reconstruction by maximum likelihood is hard.
\newblock {\em IEEE/ACM Transactions on Computational Biology and Bioinformatics}, 3(1):92--94, 2006.
\newblock \href {https://doi.org/10.1109/TCBB.2006.4} {\path{doi:10.1109/TCBB.2006.4}}.

\bibitem{solis2016inferring}
C.~Sol{\'\i}s-Lemus and C.~An{\'e}.
\newblock Inferring phylogenetic networks with maximum pseudolikelihood under incomplete lineage sorting.
\newblock {\em PLoS genetics}, 12(3):e1005896, 2016.
\newblock \href {https://doi.org/10.1371/journal.pgen.1005896} {\path{doi:10.1371/journal.pgen.1005896}}.

\bibitem{doi:10.1080/15427951.2005.10129100}
P.~P. Stephen~Boyd, Persi~Diaconis and L.~Xiao.
\newblock Symmetry analysis of reversible markov chains.
\newblock {\em Internet Mathematics}, 2(1):31--71, 2005.
\newblock \href {https://arxiv.org/abs/https://doi.org/10.1080/15427951.2005.10129100} {\path{arXiv:https://doi.org/10.1080/15427951.2005.10129100}}, \href {https://doi.org/10.1080/15427951.2005.10129100} {\path{doi:10.1080/15427951.2005.10129100}}.

\bibitem{van2022orchard}
L.~van Iersel, R.~Janssen, M.~Jones, and Y.~Murakami.
\newblock Orchard networks are trees with additional horizontal arcs.
\newblock {\em Bulletin of Mathematical Biology}, 84(8):76, 2022.
\newblock \href {https://doi.org/10.1007/s11538-022-01037-z} {\path{doi:10.1007/s11538-022-01037-z}}.

\bibitem{van2019practical}
L.~van Iersel, R.~Janssen, M.~Jones, Y.~Murakami, and N.~Zeh.
\newblock A practical fixed-parameter algorithm for constructing tree-child networks from multiple binary trees.
\newblock {\em Algorithmica}, 84(4):917--960, 2022.
\newblock \href {https://doi.org/10.1007/s00453-021-00914-8} {\path{doi:10.1007/s00453-021-00914-8}}.

\bibitem{van2014trinets}
L.~Van~Iersel and V.~Moulton.
\newblock Trinets encode tree-child and level-2 phylogenetic networks.
\newblock {\em Journal of mathematical biology}, 68(7):1707--1729, 2014.
\newblock \href {https://doi.org/10.1007/s00285-013-0683-5} {\path{doi:10.1007/s00285-013-0683-5}}.

\bibitem{van2020reconstructibility}
L.~van Iersel, V.~Moulton, and Y.~Murakami.
\newblock Reconstructibility of unrooted level-k phylogenetic networks from distances.
\newblock {\em Advances in Applied Mathematics}, 120:102075, 2020.
\newblock \href {https://doi.org/10.1016/j.aam.2020.102075} {\path{doi:10.1016/j.aam.2020.102075}}.

\bibitem{vaughan2017inferring}
T.~G. Vaughan, D.~Welch, A.~J. Drummond, P.~J. Biggs, T.~George, and N.~P. French.
\newblock Inferring ancestral recombination graphs from bacterial genomic data.
\newblock {\em Genetics}, 205(2):857--870, 2017.
\newblock \href {https://doi.org/10.1534/genetics.116.193425} {\path{doi:10.1534/genetics.116.193425}}.

\bibitem{vose1999simple}
M.~D. Vose.
\newblock {\em The simple genetic algorithm: foundations and theory}.
\newblock MIT press, 1999.
\newblock \href {https://doi.org/10.7551/mitpress/6229.001.0001} {\path{doi:10.7551/mitpress/6229.001.0001}}.

\bibitem{wen2016reticulate}
D.~Wen, Y.~Yu, M.~W. Hahn, and L.~Nakhleh.
\newblock Reticulate evolutionary history and extensive introgression in mosquito species revealed by phylogenetic network analysis.
\newblock {\em Molecular ecology}, 25(11):2361--2372, 2016.
\newblock \href {https://doi.org/10.1111/mec.13544} {\path{doi:10.1111/mec.13544}}.

\bibitem{wen2016bayesian}
D.~Wen, Y.~Yu, and L.~Nakhleh.
\newblock Bayesian inference of reticulate phylogenies under the multispecies network coalescent.
\newblock {\em PLoS genetics}, 12(5):e1006006, 2016.
\newblock \href {https://doi.org/10.1371/journal.pgen.1006006} {\path{doi:10.1371/journal.pgen.1006006}}.

\bibitem{wen2018inferring}
D.~Wen, Y.~Yu, J.~Zhu, and L.~Nakhleh.
\newblock Inferring phylogenetic networks using phylonet.
\newblock {\em Systematic biology}, 67(4):735--740, 2018.
\newblock \href {https://doi.org/10.1093/sysbio/syy015} {\path{doi:10.1093/sysbio/syy015}}.

\bibitem{wu2020inference}
Y.~Wu.
\newblock {Inference of population admixture network from local gene genealogies: a coalescent-based maximum likelihood approach}.
\newblock {\em Bioinformatics}, 36:i326--i334, 07 2020.
\newblock \href {https://doi.org/10.1093/bioinformatics/btaa465} {\path{doi:10.1093/bioinformatics/btaa465}}.

\bibitem{yang2013quartet}
J.~Yang, S.~Gr{\"u}newald, and X.-F. Wan.
\newblock Quartet-net: a quartet-based method to reconstruct phylogenetic networks.
\newblock {\em Molecular biology and evolution}, 30(5):1206--1217, 2013.
\newblock \href {https://doi.org/10.1093/molbev/mst040} {\path{doi:10.1093/molbev/mst040}}.

\bibitem{yu2014maximum}
Y.~Yu, J.~Dong, K.~J. Liu, and L.~Nakhleh.
\newblock Maximum likelihood inference of reticulate evolutionary histories.
\newblock {\em Proceedings of the National Academy of Sciences}, 111(46):16448--16453, 2014.
\newblock \href {https://doi.org/10.1073/pnas.1407950111} {\path{doi:10.1073/pnas.1407950111}}.

\bibitem{yu2015maximum}
Y.~Yu and L.~Nakhleh.
\newblock A maximum pseudo-likelihood approach for phylogenetic networks.
\newblock {\em BMC genomics}, 16(10):1--10, 2015.
\newblock \href {https://doi.org/10.1186/1471-2164-16-S10-S10} {\path{doi:10.1186/1471-2164-16-S10-S10}}.

\bibitem{zhang2018bayesian}
C.~Zhang, H.~A. Ogilvie, A.~J. Drummond, and T.~Stadler.
\newblock Bayesian inference of species networks from multilocus sequence data.
\newblock {\em Molecular biology and evolution}, 35(2):504--517, 2018.
\newblock \href {https://doi.org/10.1093/molbev/msx307} {\path{doi:10.1093/molbev/msx307}}.

\bibitem{zhang2016tree}
L.~Zhang.
\newblock On tree-based phylogenetic networks.
\newblock {\em Journal of Computational Biology}, 23(7):553--565, 2016.
\newblock \href {https://doi.org/0.1089/cmb.2015.0228} {\path{doi:0.1089/cmb.2015.0228}}.

\end{thebibliography}
